\documentclass[a4paper,11pt]{extarticle}

\usepackage[english]{babel}

\usepackage[letterpaper,top=2cm,bottom=2cm,left=3cm,right=3cm,marginparwidth=1.75cm]{geometry}
\usepackage{amsmath, amsthm}
\usepackage{amssymb}
\usepackage{enumitem}
\usepackage{graphicx}
\usepackage{color,comment}
\usepackage{stmaryrd}
\usepackage{multirow}
\usepackage{titling}
\usepackage{titlesec}
\usepackage{algorithm}
\usepackage{algpseudocode}
\usepackage{algorithmicx}
\usepackage{authblk}
\usepackage[hyphens]{url}
\algdef{SE}[EVENT]{Event}{EndEvent}[1]{\textbf{upon event}\ #1\ \algorithmicdo}{\algorithmicend\ \textbf{event}}%
\algtext*{EndEvent}
\algdef{SE}[EXISTS]{Exists}{EndEvent}[2]{\textbf{upon exists} \ #1\ \textbf{such that} \ #2\ \algorithmicdo}{\algorithmicend\ \textbf{exists}}%
\algtext*{EndEvent}
\algdef{SE}[FUNCTION]{Function}{EndFunction}%
   [2]{\algorithmicfunction\ \texttt{#1}\ifthenelse{\equal{#2}{}}{}{(#2)}}%
   {\algorithmicend\ \algorithmicfunction}%

\usepackage[colorlinks=true, allcolors=blue]{hyperref}
\usepackage{tikz}
\usepackage{titling}
\newcommand{\subtitle}[1]{%
  \posttitle{%
    \par\end{center}
    \begin{center}\large#1\end{center}
    \vskip0.5em}%
}

\makeatletter
\renewcommand\paragraph{\@startsection{paragraph}{4}{\z@}%
            {-2.5ex\@plus -1ex \@minus -.25ex}%
            {1.25ex \@plus .25ex}%
            {\normalfont\normalsize\bfseries}}
\makeatother
\title{Topos: A Secure, Trustless, and Decentralized \\ Interoperability Protocol}
\subtitle{v1.1}
\author[1]{Théo Gauthier}
\author[1]{Sébastien Dan}
\author[1]{Monir Hadji}
\author[2]{\\Antonella Del Pozzo}
\author[2]{Yackolley Amoussou-Guenou}
\affil[1]{Toposware, Inc., Cambridge, Massachusetts, USA}
\affil[2]{Université Paris-Saclay, CEA, List, Palaiseau, France}

\algnewcommand\algorithmicinput{\textbf{operation}}
\algnewcommand\Operation{\item[\algorithmicinput]}
\newtheorem{theorem}{Theorem}

\begin{document}
\maketitle

\begin{abstract}
Topos is an open interoperability protocol designed to reduce as much as possible trust assumptions by replacing them with cryptographic constructions and decentralization while exhibiting massive scalability. The protocol does not make use of a central blockchain, nor uses consensus to ensure consistent delivery of messages across a heterogeneous ecosystem of public and private blockchains, named subnets, but instead relies on a weak causal reliable broadcast implemented by a distributed network which we call \textit{Transmission Control Engine} (TCE). The validity of cross-subnet messages is ensured by the \textit{Universal Certificate Interface} (UCI) and stems from zkSTARK proofs asserting the validity of subnets' state transitions executed by the Topos zkVM. Such proofs of computational integrity are publicly verifiable by any other participants in and out the protocol such as other subnets or audit companies. The interface between the TCE and subnets leverages the ICE-FROST protocol, an innovative threshold signature scheme, whose static public key allows for uniquely identifying subnets after they register in the protocol. The Topos protocol is designed to provide \textit{uniform security} to the ecosystem and to handle any type of subnets (e.g., permissioned, permissionless) in order to fit any business use cases and pave the way for global adoption and a new standard for the Internet base layer.
\end{abstract}

\tableofcontents

\newpage

\section{Introduction}

Blockchain technology is evolving very fast. We are witnessing the development of more and more real-world applications which demonstrates strong interest from both industry and academia. Furthermore, blockchain technology is on its way to challenge the performance of centralized systems with different blockchain projects now reaching a throughput in the thousands of transactions per second, e.g., Solana \cite{Solana}, Avalanche \cite{Avalanche}, Polkadot \cite{Polkadot}, or Algorand \cite{10.1145/3132747.3132757}.

Since its inception, blockchain technology has mainly focused on creating very sparse and standalone networks, all decoupled one another, trying to solve different challenges. Such heterogeneity has forged a future of blockchain leaning towards the coexistence of multiple layer-1 chains over the domination of a single network. The proliferation of application-specific blockchains and smart contract platforms hosting new instances of existing dApps and DeFi protocols will continue to accelerate the general adoption of Web3 technologies. Thus, the need for interoperability continues to grow considerably. It now appears crucial to solve the interoperability challenge as it will improve overall blockchain scalability and pave the way for new business opportunities by composing applications hosted on different blockchain systems.

During the decade that followed the release of Bitcoin \cite{Bitcoin}, there has been a continuous and global effort to bring blockchain technology to an industrial level. This effort has primarily targeted some of the most important blockchain issues: scalability and interoperability.

\textit{Scalability}. Blockchain scalability is closely pegged to two of its upmost metrics: its latency (speed) and its throughput (capacity). Latency represents the time a transaction takes to be inserted in a block and for it to be accepted by the network, while throughput relates to the number of transactions the network is capable of adding on-chain per unit of time. For uncertain reasons, these concepts are often made misleading and thus prevent the community from gauging the true performance of a blockchain system. Transaction finality, probabilistic or deterministic, has to be considered when evaluating blockchain performance. In this context, latency is defined as the time it takes for a transaction to be finalized, while throughput is defined as the number of finalized transactions per unit of time. A later phase in blockchain technology history has seen the generalization of deterministic finality, commonly achieved by means of classical BFT consensus algorithms \cite{https://doi.org/10.48550/arxiv.2001.11965, Kwon2014TendermintC, 10.1145/2976749.2978399, https://doi.org/10.48550/arxiv.1803.05069}. Such algorithms have shown their limitations in terms of scalability for they come with a quadratic message complexity, and as such lead to much higher settlement latency as the number of validators increases.

\textit{Interoperability.} Interoperability lies in the capacity of multiple systems to interface with each other. Exchanging assets and data between blockchains is key to the adoption of the technology and yet has historically been a reality only within an environment of trust, either internally between the exchanging parties or externally via third-party bridge administrators. In addition, specialization, a scalability enabling principle centered around the segregation of a multi-app blockchain into application-specific chains, is a straightforward call for interoperability.

For an industrial age of blockchain to emerge, we add two additional components: \textit{composability} and  \textit{privacy}.

\textit{Composability.} Composability is a design principle that allows different components within a system to be combined to meet any specific use case requirements. Within a single blockchain network like Ethereum \cite{Ethereum}, composability is atomic: Smart contract functions can invoke other contracts synchronously with the insurance that either all contract calls succeed or none does. In a context of cross-chain interoperability, composability is obtained when business logic deployed on different blockchains can interact with each other to create new value.

\textit{Privacy.} Historically, privacy has remained a rarity in the blockchain scene, transactional data being accessible to all network participants by design even though private transaction protocols such as \cite{Monero, Zcash} have allowed for keeping this data hidden while preserving transaction validity. For enterprise use, being interoperable while keeping internal data private is fundamental. Such a model is still lacking today, preventing organizations from switching from Web1-2 technologies to Web3 ones.

\subsection{Our Contributions}

We propose Topos, a generalized interoperability protocol designed for transmitting messages across sovereign blockchains. The Topos ecosystem is composed of a permissionless reliable broadcast primitive \cite{10.5555/1972495} and a scalable set of decoupled public and private blockchains, named subnets. Topos ensures the validity of state transitions without relying on fraud proofs \cite{AlBassam2018FraudAD} or on designated subsets of participants to perform validity checks \cite{Polkadot}.

In the interoperability landscape, trustlessness is defined as the absence of trust in the interoperability protocol itself and relying instead on the security of the underlying blockchains. While this is an improvement over trusted interoperability protocols, this does not permit complete trustlessness and cannot provide a higher level of security than that of the interoperated blockchains. With Topos, we decouple the validity of cross-chain transactions from the security of the underlying blockchains by replacing this coupling with a zkSTARK proof system, providing irrefutable evidence of the validity of messages.

As part of their implementation of the \textit{Universal Certificate Interface} (UCI), subnets integrate the Topos zkVM (Zero-Knowledge Virtual Machine) to host an arbitrary number of applications which can exchange assets and arbitrary data with other subnets by exchanging objects called certificates. Certificates are a central component used by subnets to exchange cross-subnet messages. Cross-subnet message transmission is handled by the \textit{Transmission Control Engine} (TCE), a decentralized network implementing a reliable broadcast primitive. Additionally, the protocol enforces computational integrity of all state transitions among all subnets by means of a zkSTARK proof system. Moreover, we introduce modifications to the FROST signature scheme \cite{frost} to allow the TCE participants to authenticate incoming certificates prior to their verification and delivery. The combination of UCI and TCE provides the ecosystem with \textit{uniform security} and as such, subnets do not need to rely on any trust assumptions but cryptographic assumptions for cross-subnet message passing. Finally, we present the Topos Subnet, a subnet responsible for maintaining registration on the protocol in order to minimize overall protocol complexity, while not being used for state synchronization nor cross-subnet message passing.

The rest of the paper is organized as follows. Related works are discussed in Section \ref{sec:related-work}. In Section \ref{sec:properties}, we introduce the properties of our solution. We present in Section \ref{sec:design} the design considerations of our protocol and its components. Section \ref{sec:use-cases} shows several use cases in which our solution is well-suited. In Section \ref{sec:conclusion}, we summarize the paper with concluding remarks and Section \ref{sec:discussion} is dedicated to additional discussions and future works for our solution.

\subsection{Related Work}\label{sec:related-work}

Over the years, several projects have focused on interoperability and scalability. Here we introduce some of the most contributing projects. \\

\textbf{Cosmos.} \cite{Cosmos} Cosmos is a network of sovereign blockchains called zones. Zones are endowed with Tendermint \cite{Kwon2014TendermintC} and are connected with each other via the IBC protocol. Cosmos can only support interoperability between BFT-based blockchains and to do so employs a simple model for their inter-blockchain communication protocol by centering its design around the use of decentralized relayers and on-chain light-clients to allow connected zones to verify each other's block headers and validate transaction inclusion proofs. The validity of cross-chain transfers is left to interpretation based on the trust zones have in each other since, as opposed to Topos, Cosmos does not use validity proofs of state transitions.\\

\textbf{Polkadot.} \cite{Polkadot} Polkadot is a shard system composed of a central entity called the Relay Chain along with shards called parachains. The purpose of Polkadot is to let parachains free of any security concerns so that developers can focus on the application layer. This is enabled by two factors (a) an abstraction of the internals of the parachain providing a standard verification for all parachains (b) an active validation executed by randomly and frequently sampled Relay Chain actors. These properties reflect the so-called \textit{Shared Security}. The validity of parachain block candidates is ensured by Relay Chain validators, hence prevents parachains from being sovereign networks and from having a private state. On the contrary, Topos allows subnets to have their own consensus, a private state, and trustless interoperability.\\

\textbf{Avalanche.} \cite{Avalanche} Avalanche is a highly scalable blockchain platform that focuses on the deployment of blockchains with three main targets: application-specific blockchains, smart contract platforms, and digital asset platforms. In Avalanche, blockchains are deployed within subnets (which are sets of validators). Validators are validating all chains in their subnet, as well as the three chains composing the Primary Network. The Avalanche platform has no bounds in the number of blockchains which can participate but offers interoperability only between the chains of the same subnet. Topos does not face this limitation: Upon implementing the Topos protocol, subnets are interoperable with the whole ecosystem. Although cross-subnet interoperability is envisioned on Avalanche, no design has been revealed yet.\\

\textbf{Chainlink.} \cite{Chainlink} Chainlink 2.0 is a framework that aims at solving the oracle problem \cite{Oracle} by introducing the Decentralized Oracle Network (DON). Historically, oracle services introduce trust. However, Chainlink tackles this problem by filtering the off-chain data source through a BFT layer. The committee of Oracles that composes the DON sign their reports with a multi-signature scheme. By doing so, Chainlink is increasing decentralization and minimizing trust in oracle services. Furthermore, Chainlink proposes an interoperability feature with its Cross Chain Interoperability Protocol (CCIP). The global consistency of cross-chain communication in CCIP is reduced to the security of their Anti-Fraud Network which is a dedicated DON actively watching for misbehavior across other DONs. In Topos, global consistency of cross-subnet messages is passively obtained via the TCE's reliable broadcast protocol. In addition, validity of cross-subnet messages is ensured cryptographically in Topos whereas Chainlink relies on agreements between oracles.\\

\textbf{LayerZero.} \cite{LayerZero} LayerZero is an interoperability protocol which decouples provision of block headers and transaction proofs to allow for trustless cross-chain communication. Each action is handled by two centralized or decentralized parties, namely an Oracle and a Relayer. The cross-chain message passing is trustless under the assumption of independence between the Oracle and the Relayer. The Oracle transports the block headers while the Relayer submits the transaction inclusion proofs. In LayerZero, cross-chain transactions are not cryptographically guaranteed to be valid as opposed to Topos's proofs of computational integrity which enforce validity of all transactions.\\

\section{Properties}\label{sec:properties}

An interoperability protocol should be trustless, secure, and have strong network effect. The following properties need to be maximized if the vision of an “Internet of blockchains” is to be realized. By design, the Topos protocol fulfills these properties comprehensively.\\

\textbf{Trustless.} Subnets receiving certificates and cross-subnet messages from another subnet should have guarantees as to the validity of these cross-subnet messages. These guarantees should not rely on trust assumptions in centralized entities, decentralized actors, or the interoperated subnets, but on cryptographic assumptions. Leveraging succinct zero-knowledge proofs allows for removing this trust completely from the equation and solely relying on mathematical truth.\\

\textbf{Security.} The protocol must be robust and prevent an adversary from creating conflicting certificates in an attempt to double-spend via cross-subnet messages, as it would cause consistency issues in the system.\\

\textbf{Scalability.} There should be no limit to the number of participants in the ecosystem. The protocol should be able to handle an arbitrarily large number of subnets, as well as to seamlessly scale to millions of TCE participants, by ideally ensuring logarithmic communication complexity per participant. Furthermore, the system must have very high capacity to be able to process a massive amount of cross-subnet messages.\\

\textbf{Privacy.} It should be possible for subnets to keep their internal state hidden from the rest of the ecosystem. Thus, the protocol cannot rely on having the receiving subnet nor any third-parties actively verify cross-subnet messages by accessing the state of the sending subnet. Instead, cross-subnet messages should contain indisputable evidence that these messages are correct. By design, the protocol should be able to handle any type of subnets, i.e., public and private subnets.\\

\textbf{Authentication.} It is important that data exchanged between subnets is authenticated to provide guarantees of integrity. Authentication using threshold signatures should have a public key that remains static for the whole lifespan of subnets to facilitate key management. As such, it should be possible for any actors, in and out of the protocol, to verify authenticity of cross-subnet messages.\\

\textbf{Decentralization.} The protocol should allow for permissionless participation in the TCE and open registration of subnets. Participation should not be handled by a central authority, and processes should be able to join the system at any time. To enable high levels of decentralization, it is also necessary that the entry cost for participation remains low, such that common hardware is enough to fully participate in the system.\\

\section{Design}\label{sec:design}

The Topos protocol is a generalized interoperability protocol which enjoys strong network effect. Once a blockchain has implemented Topos, it becomes interoperable with all the blockchains in the ecosystem, without any overhead. Topos complies with all the properties detailed in the previous section. In this section, we will first describe all the components that compose Topos, then detail the protocol itself.

\subsection{System Overview}

Here, we define all the components that together make the Topos protocol.

\subsubsection{Subnets}

Subnets are sovereign blockchain networks which implement the Topos protocol, devise their own consensus rules, and control their own native asset. New subnets join the ecosystem to be natively interoperable with all existing subnets without making any compromise on their sovereignty and without the need to trust any middleman. Though not a protocol requirement, subnets are expected to implement classical BFT protocols to enforce deterministic finality. This will help the subnet to guarantee that the state submitted to the rest of the ecosystem is finalized, i.e., cannot be reverted, hence avoiding the subnet to be inconsistent between its internal state and its submissions.

The first subnet client will be our Edge DevKit which is Topos's extension of the Polygon Edge framework \cite{Edge}. As any Edge native blockchain, subnets can implement their own consensus protocol. Topos's Edge DevKit additionally adds on top of Edge the necessary components for subnets to be compatible with the UCI. One significant addition is the integration of the Topos zkVM (see \ref{subsubsec:zkVM}) as the core smart contract execution environment for subnets in the Topos ecosystem.

In later iterations of the protocol, other DevKits will be created by the Topos community as extensions of other blockchain frameworks and will allow developers familiar with any tech stack to join the Topos ecosystem.

\subsubsection{zkVM}
\label{subsubsec:zkVM}

Subnets implement the Topos zkVM, a zero-knowledge virtual machine that exposes a Turing complete programming language which allows instructions to be provable with zero-knowledge proofs. dApp developers can use the Topos zkVM programming language to write any type of application that are deployed on any subnets in the form of smart contracts whose executions are provable. Developers can as such leverage the composability offered by the Topos protocol by composing their applications with other zkVM-compatible applications deployed on any subnets in the Topos ecosystem.

The Topos zkVM has been conceived to offer a set of instructions efficiently verifiable with our zkSTARK proof system. This instruction set, while small and simple, remains expressive enough for developers to easily write any kind of application on subnets. We also include into the default instruction set additional operation-specific instructions (e.g., range check, curve point addition, hash evaluation, etc.), to allow programmers to execute common operations directly without the burden of writing them with the original instruction set. The Topos zkVM execution remains extremely fast to verify—maintaining the overall scalability of the system—even when extending the original instruction set with custom complex ones.

\subsubsection{Universal Certificate Interface}
\label{subsubsec:uci}
The Universal Certification Interface (UCI) enfolds the concept of proving and verifying data across different subnets. This notion is key in trustless and secure interoperability: A sending subnet generates data intended for another subnet and the receiving subnet is ensured that the data is valid and authentic without the need for trust in the sending subnet or any third party. The UCI offers an abstraction of the internal structure of a subnet to guarantee these properties of validity and authentication without compromising the sovereignty and privacy of the sending subnet.

In the Topos ecosystem, the UCI exposes the interface that all subnets implement in order to be interoperable with each other, i.e., exchange certified data. This interface describes how certificates are to be constructed and authenticated by subnets.

\paragraph{Certificate}

A certificate is an authenticated object that wraps exchanged data with a proof of validity. The data subnets exchange are cross-subnet messages, i.e., cross-subnet asset transfers and remote arbitrary smart contract calls. Authentication of certificates is done with Topos's ICE-FROST signature scheme and proofs of valid state transition are created using our zkSTARK proof system.\\

The structure of a certificate is described below:
\begin{itemize}
    \item \verb|subnet_id| is the static ICE-FROST public key used as the unique subnet identifier;
    \item \verb|prev_state_hash| is the previous subnet state commitment (from the previous certificate);
    \item \verb|state_hash| is the current subnet state commitment;
    \item \verb|proof| is the zkSTARK proof of validity;
    \item \verb|XS_list| represents the list of included cross-subnet messages;
    \item \verb|proof_XS_list| is the list of inclusion proofs of cross-subnet messages in the proven state transition.
\end{itemize}

By including proofs of valid state transition in certificates, sending subnets prove the validity of all their internal transactions (including cross-subnet messages) executed since their previous certificate. This allows a receiving subnet, i.e., a subnet to which at least one cross-subnet message contained in the certificate is addressed, to verify the validity of a message without having access to the state of the sending subnet nor relying on a third party to verify the complete state transition. We envision that the actors creating these proofs will be subnet validators although the Topos protocol does not impose any requirements.\\

A valid state transition is defined as follows:\\

\textit{Definition (Valid State Transition).} Let $\texttt{STF}: \mathcal{S}_k \times \mathcal{T}_{k+1} \rightarrow \mathcal{S}_{k+1}$ be a state transition function, where $\mathcal{S}_k$ is the $k$-th subnet state committed to in the $k$-th certificate, and $\mathcal{T}_{k+1}$ is a set of transactions which applied to $\mathcal{S}_k$ results in $\mathcal{S}_{k+1}$. We say that a state transition is valid if and only if: $\forall \ tx \in \mathcal{T}_{k+1}$, $tx$ is a transaction correctly executed by the Topos zkVM.\\

The zkSTARK proof included in the certificate verifies that the set of transactions $\mathcal{T}_{k+1}$ between $Cert_k$ and $Cert_{k+1}$ is a valid state transition. While this does not ensure the validity of the subnet state, it guarantees the validity of its state transitions. Thus, if the state initially committed to by the subnet as per its registration and all of its subsequent state transitions are valid, then by transitivity its latest state is valid.\\

The certificate validation is handled by the $\texttt{Valid\_cert}$ predicate, defined in Algorithm \ref{alg:valid_cert}, which calls the zkSTARK $\texttt{Verif}_\mathcal{C}$ predicate (see Equation \ref{eq:verif_c}) on the certificate data to assert the validity of the committed state transition, checks the inclusion proofs of the cross-subnets messages in the proven state transition, and returns {\sf true} if both checks succeed. In this case, the predicate provides the certificate with \textit{intrinsic validity}: The certificate contains all the necessary information to prove its validity and its verification does not depend on an external state—the predicate is stateless and so trivially monotonic.\\

\begin{algorithm}
\caption{Certificate validation predicate}\label{alg:valid_cert}
\begin{algorithmic}

\Function{Valid\_cert}{$Cert$}
    \State \Return $\texttt{Verif}_\mathcal{C}(Cert.\verb|proof|, Cert.\verb|prev_state_hash|, Cert.\verb|state_hash|) \ \wedge $
    \State \hspace{11mm} $\ \texttt{Verify\_incl}(Cert.\verb|proof_XS_list|) $
\EndFunction

\end{algorithmic}
\end{algorithm}

\paragraph{ICE-FROST Signature}

A signature scheme $\mathsf{TS = (KeyGen,Sign,Verify)}$ with key generation, signing, and verification algorithms $\mathsf{KeyGen}$, $\mathsf{Sign}$, and $\mathsf{Verify}$ respectively, and security parameter $\lambda$ is a $(t,n)$  threshold signature scheme if the following conditions hold:
\begin{itemize}
    \item \textit{Correctness.} Any subsets of participants with cardinality at least $t$ can produce a valid signature on message $m$. A valid signature is a signature that will be verified by the $\mathsf{Verify}$ algorithm.
    \item \textit{Unforgeability.} Any polynomial-time adversary who can corrupt up to $t-1$ players and views the protocol output (signature) on $\mathsf{poly}(\lambda)$ input messages of their choice cannot produce the valid signature $\sigma$ for a message $m$ that has not been submitted to the $\mathsf{Sign}$ algorithm before.
\end{itemize}
The Topos protocol employs threshold signatures to authenticate certificates, i.e., allow actors of the ecosystem to verify that a propagated certificate has been created by the correct subnet and has not been tampered with in transit. The Topos ICE-FROST signature \cite{cryptoeprint:2021:1658} is the first to consider static private/public keys for a round-optimized Schnorr-based signature scheme \cite{10.1007/0-387-34805-0_22}. With static public keys, the group's established public and private keys remain the same for the lifetime of the subnets, while the signing shares of each participant are updated over time, as well as the set of group members. This ensures the long-term security of the static keys and facilitates the verification process of the generated threshold signature because a group of signers communicate their public key to the verifier only once during the subnet's lifetime.

Dealerless threshold signature schemes usually need to run a Distributed Key Generation (DKG) \cite{DKG} protocol each and every time the set of participants changes, resulting in a new public key. However, the TCE requires knowing the public key associated with the signature in order for the TCE participants to verify the signatures applied to certificates. A natural approach would be to include the threshold signature public key for the next certificate in the current certificate but such short-lived public keys clearly lead to large overhead and are not suited for the Topos protocol.

Our contribution to the field of threshold signatures makes long-lived static public keys possible. Topos uses ICE-FROST \cite{cryptoeprint:2021:1658} to enforce usage of a single static key for the whole lifespan of the subnet, no matter how many times the set of validators changes. This allows for a lighter and simpler subnet key management.

In order to use a long-lived public key for each subnet, we add a share update property to our scheme. To update the shares for each validator set, participants secretly share the value ``0" and send corresponding shares to other participants. These new shares are added to previous shares to randomize them without changing the value of the shared secret. Randomization of shares guarantees unforgeability of the threshold signature scheme against a static adversary, i.e., an adversary who can corrupt up to $t-1$ participants. A dynamic adversary on the other hand can corrupt different $t-1$ participants in each validator set. Because validators secret shares need to be encrypted when redistributed, we need the additional property of forward secrecy. That is, an attacker that would get access to some validator decryption key would only be able to derive decryption keys between this compromised validator and future ones, but would not be able to decrypt messages encrypted and shared by previous validator sets during the shares redistribution phase, and hence would not gain knowledge of additional secret shares. This key property ensures that even if the adversary corrupts different subsets of participants in consecutive validator sets, they still cannot forge a valid signature.

\paragraph{STARKs}

When a subnet submits a certificate, it commits to a new state and certifies the state transition. The new state is the state committed to in the previous certificate on which the state transition is applied. To convince other subnets that the state transition is valid and consistent with the previous state, the certificate contains a proof of computational integrity, a zkSTARK proof \cite{BenSasson2018ScalableTA}.

A STARK proof guarantees that a computation has been correctly executed and has returned a certain output, and (if needed) without revealing the input. For example, a STARK proof can guarantee that:
\begin{itemize}
    \item The state $\mathcal{S}_{k+1}$ is the state $\mathcal{S}_k$ plus some transactions, without revealing the transactions.
    \item The hash of the state $\mathcal{S}_{k+1}$ is the hash of the state $\mathcal{S}_k$ plus some transactions, without revealing the transactions nor the states.
    \item An Account $a$ on subnet $S_i$ made a valid (holds enough funds and signed) transfer of $x$ tokens to an account $b$ on subnet $S_j$, without revealing the balance of $a$.
\end{itemize}

More formally, a STARK proof is sent by a prover $\mathcal{P}$ to convince a verifier $\mathcal{V}$ that it ran a certain computation $\mathcal{C}$ with some input $\mathcal{I}$ (and possibly obtained some output $\mathcal{O}$). The STARK system is made of a proving algorithm and a verifying algorithm. While $\mathcal{C}$ is known to both $\mathcal{P}$ and $\mathcal{V}$, $\mathcal{I}$ and $\mathcal{O}$ could be partially or fully kept secret by the prover or shared between the prover and the verifier, depending on the statement to be proven.
The entire process takes four steps:

\begin{enumerate}
    \item $\mathcal{P}$ runs $\mathcal{C}$ with input $\mathcal{I}$ and records an execution trace $t$. Broadly speaking, the trace is a 2D-matrix recording the value of all the variables of $\mathcal{C}$ at each execution step. $\mathcal{P}$ also saves the output $\mathcal{O}$ if any.
    \item $\mathcal{P}$ executes the proving algorithm $\texttt{Prov}_\mathcal{C}$ on input $t$ (and $\mathcal{O}$ if any), which returns a STARK proof $\pi$ that $\mathcal{C}$ has been correctly executed with some input, and returned $\mathcal{O}$ as an output (if any).
    \item $\mathcal{P}$ sends $\pi$ (and possibly parts of $\mathcal{I}$, $\mathcal{O}$, or functions of them, depending on the statement) to $\mathcal{V}$.
    \item $\mathcal{V}$ executes $\texttt{Verif}_\mathcal{C}$ on input the proof and potential inputs/outputs received. It returns $\textsf{true}$ if the proof has been computed from a valid execution of $\mathcal{C}$ that returns $\mathcal{O}$ (if any) on input $\mathcal{I}$, and $False$ otherwise.
\end{enumerate}

STARK systems are known to be doubly scalable, with a prover that is running in time $O(t \log^2 t)$, dominated by the Fast-Fourier-Transform interpolation, and a verifier that scales in only poly-logarithmic time with the trace length, i.e., $O(\log^2 t)$. This allows subnets to prove exponentially large computations and hence improve the overall scalability of the Topos protocol. In addition, such systems are post-quantum secure, as only relying on symmetric primitives like hash functions, unlike their SNARK \cite{184425, Groth16, Plonk} counterparts based on asymmetric primitives.

However, one issue with the above process is that $\texttt{Prov}_\mathcal{C}$ and $\texttt{Verif}_\mathcal{C}$ both depend on $\mathcal{C}$. 
In other words, a distinct pair of proving and verifying algorithms is needed for each specific computation.
Not only does it require both participants to potentially store and execute multiple algorithms, but it also forces $\mathcal{P}$ to write a specific proving algorithm for every new computation, for example for a new smart contract, it wants to prove. 
$\mathcal{V}$ would likewise need to make sure to keep its verifying algorithm up to date. 
For these reasons, we adopted a general-purpose approach: Our STARK system can prove arbitrary computations with a single pair of proving and verifying algorithms that do not need to be updated if the program to prove is modified.

More precisely, the computation $\mathcal{C}$ which is proven is the Topos zkVM execution itself. 
$\mathcal{I}$ consists in a state hash and all the operations happening on-chain modifying this state.
$\mathcal{O}$ is the final state being returned by the Topos zkVM after applying the provided input state transition on the input state.
After a proof has been computed, $\mathcal{P}$ sends it to $\mathcal{V}$, along with a hash of the final state $\mathcal{O}$. The hash of the previous state can be retrieved from the latest verified certificate of the subnet $\mathcal{P}$ belongs to. Only providing the state hashes allows sending subnets to keep their state private and improves on scalability by reducing the overhead in data transmitted. The verification function is defined as follows:
\begin{equation}\label{eq:verif_c}
    \texttt{Verif}_\mathcal{C}(\pi_k, Hash(\mathcal{S}_{k-1}), Hash(\mathcal{S}_k)) =
    \begin{cases}
    $\textsf{true}$, & \text{if the proof is valid}\\
    $\textsf{false}$, & \text{otherwise}
    \end{cases}
\end{equation}

The verification function attests to the validity of the state transition claimed by a sending subnet, from a previous state that was committed to (in the form of a hash), to a new committed state. This algorithm can only output $\textsf{true}$ if the prover submitted a valid state transition as part of its (private) input $\mathcal{I}$, i.e., corresponding to valid executions of the Topos zkVM.\\

Since the Topos protocol relies on cryptography for subnets to prove the validity of their state transitions to the rest of the ecosystem, it is crucial to have efficient cryptographic primitives in order to preserve high scalability. STARK proof systems–in particular the ones based on FRI (see \ref{stark:fri}), as is the case for the one used in our protocol–have very light requirements (namely to work on a prime field $\mathbb{F}_p$ with a $2^k$\textit{-th} root of unity for relatively large $k$) whereas other common SNARK constructions are based, among other things, over algebraic groups, which involve complex mathematical operations that can be hard to implement and optimize. In particular, the prime field involved in STARK proof systems can be much smaller than the usual cryptographic size of 256 or 512 bits.
However, blockchains always require some digital signature scheme to assert the authenticity of propagated messages and, to date, digital signatures based on elliptic curves, such as Schnorr signatures, EdDSA or ECDSA are the preferred ones, due to both their speed and resulting size. The underlying curves commonly used in pair with those schemes are all of large cryptographic sizes, and hence prevent from benefiting fully from the mathematical structure of our proving system.
To address this, and to offer subnets the possibility to exploit the whole power of STARKs, we designed a new elliptic curve, Cheetah \cite{Cheetah}, constructed over a sextic extension of a small field with characteristic $p = 2^{64} - 2^{32} + 1$ and tailored for efficiency when proving operations over its group. A detailed security analysis and description of the deterministic process that generated this curve is available at \cite{cryptoeprint:2022:277}. With Cheetah, the Topos zkVM execution $\mathcal{C}$ can be proven while maintaining a small proof system base field, a crucial consideration for the efficiency of the protocol.\\

The STARK system at the core of Topos enables the protocol to be:

\begin{itemize}
    \item Trustless: The soundness property of STARKs ensures that it is computationally infeasible for a malicious prover to create a valid proof for an invalid statement. This means that validity of state transitions solely depends on the soundness of the STARK proof included in certificates.
    \item Private: Instead of providing the whole computation that updated their internal state to the verification function, subnets pass only the known hash of their previous state along with the hash of their new state, thus do not reveal anything about transactional data. The computational integrity ensured by the STARK proof system combined with zero-knowledge guarantees that no additional information about the state of subnets is revealed to verifiers; this grants full privacy to subnets.
    \item Scalable: STARKs can prove the computational integrity of a very large number of transactions while keeping the verification cost extremely small.
\end{itemize}

\subsubsection{Topos Subnet}

The Topos Subnet is a blockchain network whose main purpose is to handle registration of the ecosystem actors, namely the subnets and the TCE participants, to manage TOPOS, the ecosystem's native cryptocurrency, and to allow for governance of the protocol through on-chain voting, such that TOPOS token holders will have the ability to participate in future protocol improvements. Subnets register themselves by sending a special transaction which pays a dedicated fee denominated in TOPOS. Furthermore, the Topos Subnet is leveraged for the Sybil resistance of the TCE, requiring participants to lock a TOPOS amount in order to join the system. Finally, it enables the setup of an incentive mechanism for the TCE participants to be rewarded when following the prescribed protocol.\\

As for the actual implementation, the Topos Subnet is built with the Edge framework \cite{Edge} and implements the IBFT \cite{https://doi.org/10.48550/arxiv.2002.03613} consensus. IBFT ensures deterministic finality, guaranteeing that blocks can never be reverted once finalized unlike protocols where finality is only probabilistic. Through the process of nomination and validation, an unbounded number of TOPOS token holders are economically incentivized to participate and contribute to the security of the system.

\subsubsection{Transmission Control Engine}

As seen in Section \ref{subsubsec:uci}, the UCI ensures that subnets' state transitions are valid (guaranteed by the STARK proof) and that the certificates transporting them are authenticated (guaranteed by the ICE-FROST signature). To allow for trustless cross-subnet communication, subnets additionally rely on the Transmission Control Engine (TCE), a network of nodes that receives certificates submitted by subnets to consistently deliver them, i.e., prevent subnets from  having conflicting certificates successfully processed.

The TCE implements a permissionless probabilistic protocol of causal reliable broadcast based on \cite{Guerraoui2019TheCN}. The protocol does not involve consensus since consensus enforces total ordering on messages while it is sufficient to have causal ordering for our purposes, i.e., certificates from the same subnet do not commute, while two independent certificates from two different subnets commute. Causal ordering is needed to make sure that the protocol processed all dependencies of a specific certificate as shown in Figure \ref{fig:partial_ordering}. This results in a simpler, more efficient and more robust protocol than consensus-based solutions.

A key role of the TCE is to deal with the situation where a subnet is under attack or is controlled by an adversary, and tries to double-spend. A subnet $S_m$ controlled by an adversary may send the same assets twice to different subnets $S_i$ and $S_j$, i.e., $S_m$ sends two conflicting $n$-th certificates ($Cert_{n}$ to $S_i$ and $Cert_{n'}$ to $S_j$). In other words, and as shown in Figure \ref{fig:conflicting_cert}, for two certificates $Cert_{n}$ and $Cert_{n'}$, they are said to be conflicting if both $Cert_{n}$ and $Cert_{n'}$ are valid with respect to $Cert_{n-1}$ but the operations associated with the two certificates do not have a legal sequential history. Without a mechanism to prevent conflicting certificates, $S_i$ and $S_j$ would execute messages on-chain from $Cert_{n}$ and $Cert_{n'}$ respectively, in which case $S_m$ would successfully be able to double-spend.\\

\begin{figure}
\centering
\includegraphics[width=0.8\textwidth]{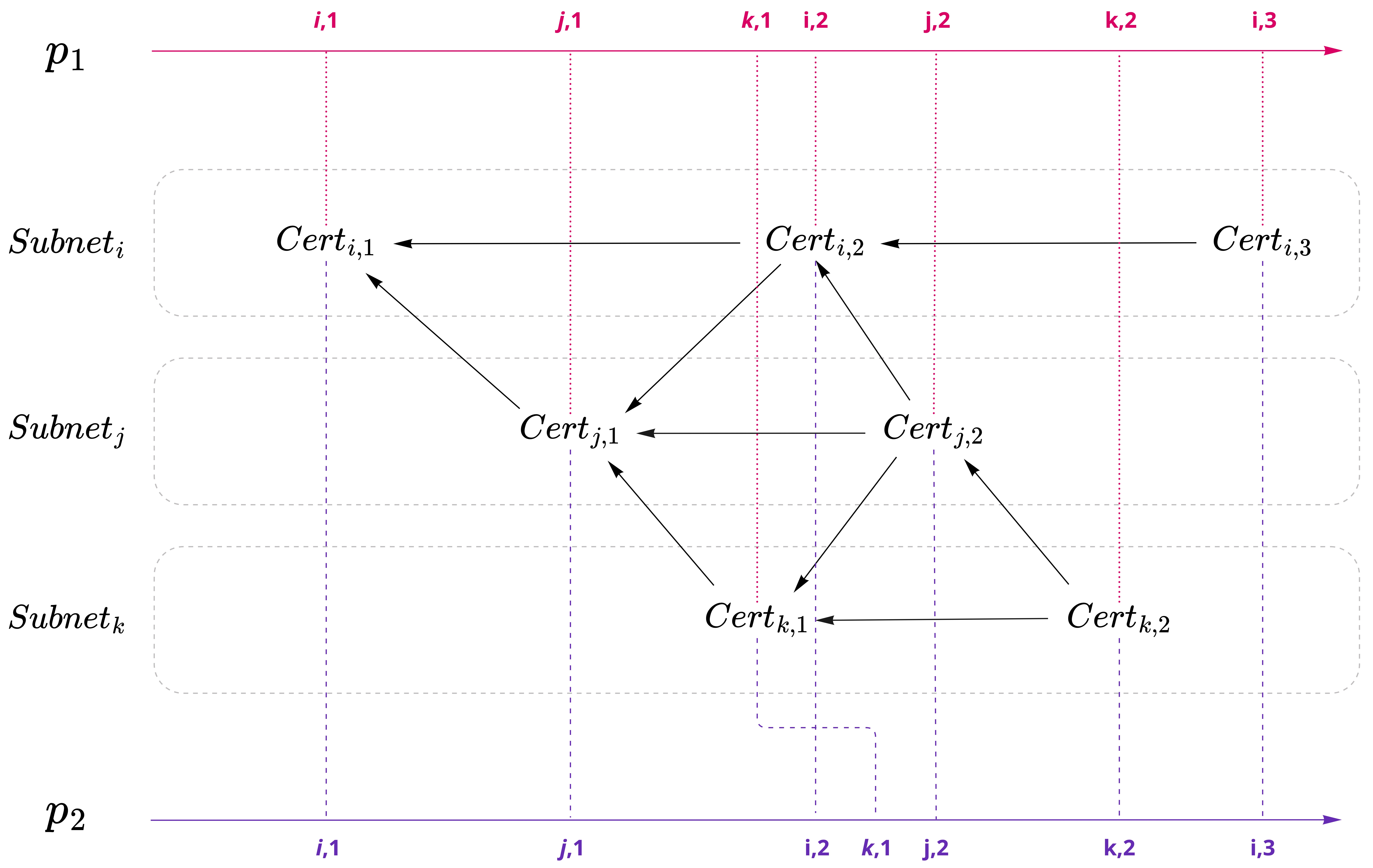}
\caption{\label{fig:partial_ordering}TCE partial ordering of cross-subnet messages. TCE processes $p_1$ and $p_2$ maintain the same set of certificates but ordered differently.}
\end{figure}

\begin{figure}
\centering
\includegraphics[width=0.8\textwidth]{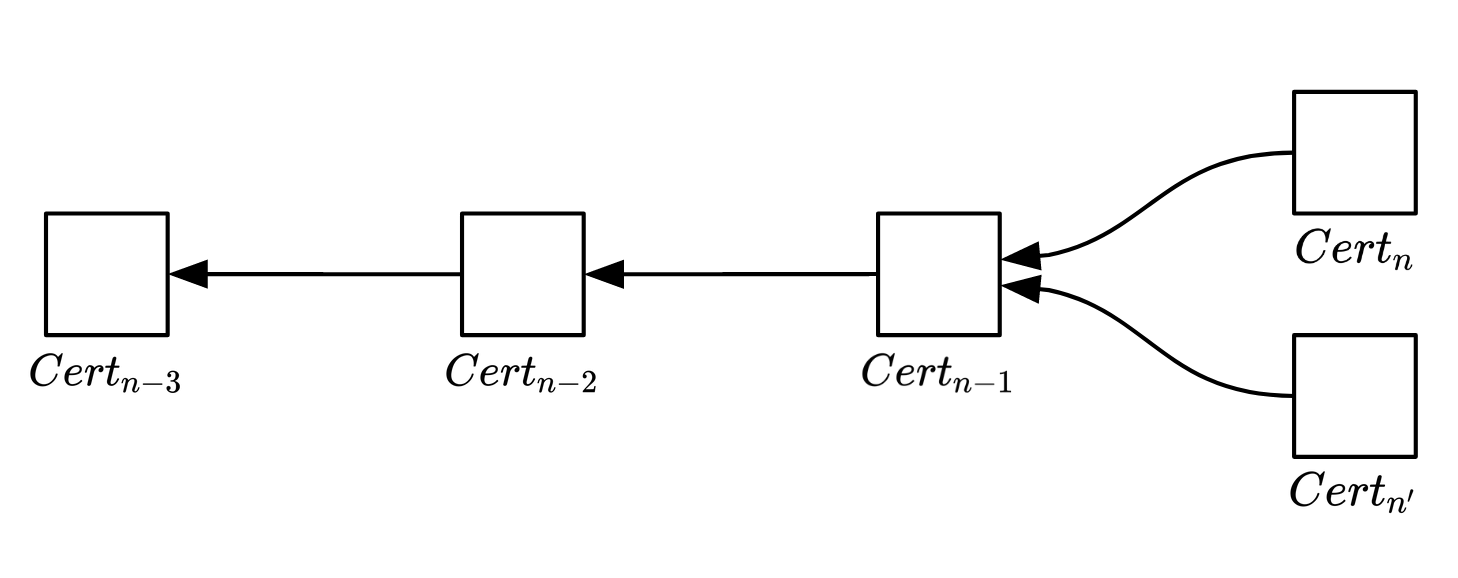}
\caption{\label{fig:conflicting_cert}$Cert_n$ and $Cert_{n'}$ are both valid with respect to the $\texttt{Valid\_cert}$ function but conflict with each other.}
\end{figure}

\paragraph{System Model}

Here we briefly describe the system model for the TCE reliable broadcast. We consider a set $\Pi$ of $n$ processes. Processes are equipped with private and public keys and identified with the latter.
At most a fraction $f$ of all processes is Byzantine, i.e., subject to arbitrary failures. We say that a process is correct if it follows the prescribed protocol. Byzantine processes cannot determine which correct processes another correct process is communicating with. Byzantine processes are under the control of the same adversary, and can take coordinated actions.

Processes can communicate with each other using the Probabilistic Reliable Broadcast (PRB) primitive  \cite{AT2,Guer}. An instance of probabilistic reliable broadcast exports two events:
\begin{itemize}
    \item $prb$.Broadcast($m$): Used by a process inside the system to broadcast a message $m$;
    \item $prb$.Deliver($p,m$): Used by a process inside the system to handle the delivery of a message $m$ from sender $p$.
\end{itemize}
For any $\epsilon \in [0, 1]$, we say that the protocol implementing the reliable broadcast is $\epsilon$-secure if the following properties hold:
\begin{itemize}
    \item No duplication: No correct process $prb$.Delivers more than one message.
    \item Integrity: If a correct process $prb$.Delivers a message $m$, and the sender $p$ is correct, then $m$ was previously $prb$.Broadcast by $p$.
    \item $\epsilon$-Validity: If the sender $p$ is correct, and $p$ $prb$.Broadcasts a message $m$, then $p$ eventually $prb$.Delivers $m$ with probability at least (1 - $\epsilon$).
    \item $\epsilon$-Consistency: Every correct process that $prb$.Delivers a message $prb$.Delivers the same message with probability at least $(1 - \epsilon)$.
    \item $\epsilon$-Totality: If a correct process $prb$.Delivers a message, then every correct process eventually $prb$.Delivers a message with probability at least $(1 - \epsilon)$.
\end{itemize}

\paragraph{Probabilistic Reliable Broadcast}\label{probaRB}

For completeness, we recall the probabilistic reliable broadcast solution as presented in \cite{AT2}, to underline its advantages and later detail the modifications we apply to it.

The probabilistic reliable broadcast protocol at the heart of the TCE is asynchronous, permissionless and tolerant to Byzantine (arbitrary) failures. It replaces classical quorums by stochastic samples which do not need to intersect and are of much smaller size compared to quorums. The protocol has $O(\log n)$ message complexity per process—thus an overall message complexity of $O(n \log n)$—and the sample size per process is logarithmic in the size of the system which hints at its massive scalability capabilities. The properties of Byzantine reliable broadcast of the TCE can be violated with probability $\epsilon$, which can be made arbitrarily small.

The algorithm is composed of three sequential phases of communication exchanges, namely, subscriptions, $Echo$ exchanges, and $Ready$ exchanges.
Upon initialization, a correct process randomly samples three sets via a uniform random oracle. The size of each set and the associated threshold are security parameters of the protocol. At the end of the initialization, each correct processes has the following:
\begin{itemize}
    \item An $Echo$ sample $\mathcal{E}$ and its associated threshold $E < \vert \mathcal{E} \vert$
    \item A $Ready$ sample $\mathcal{R}$ and its associated threshold $R < \vert \mathcal{R} \vert$
    \item A $Delivery$ sample $\mathcal{D}$ and its associated threshold $D < \vert \mathcal{D} \vert$
\end{itemize}

During the subscription phase, a process starts exchanging \textit{sample-specific} subscription messages with processes in the $Echo$, $Ready$ and $Delivery$ sets.
Upon receiving \textit{sample-specific} subscription messages from other processes, a process adds the corresponding message senders in new samples in the following manner:
\begin{itemize}
    \item Senders of $Echo$ subscription messages are added to a new $\tilde{\mathcal{E}}$ set.
    \item Senders of $Ready$ subscription messages are added to a new $\tilde{\mathcal{R}}$ set;
\end{itemize}

These sets determine how processes communicate between them. A process interacts with its sets as follows:
\begin{itemize}
    \item It only listens for messages coming from peers in $\mathcal{E}$, $\mathcal{R}$ and $\mathcal{D}$ sets;
    \item It only sends messages to peers in $\tilde{\mathcal{E}}$ and $\tilde{\mathcal{R}}$ sets.
\end{itemize}

We note that each set is such that its size is $\Omega(\log n)$.\\

Now we proceed with describing the algorithm implementing the probabilistic reliable broadcast. Processes can communicate directly with each other or broadcast messages using a probabilistic (not reliable) broadcast primitive, which might deliver conflicting copies of a message. In the following we refer to it as $pb$.Broadcast and $pb$.Delivery. 

Upon the invocation of $prb$.Broadcast for a message $m$, a process $p$ $pb$.Broadcasts it to all the processes in $\Pi$. 

Upon $pb$.Delivering a correctly signed message $m$, a correct process sends an $Echo$ message to every process in its $\tilde{\mathcal{E}}$ subscription set. 
\\

A correct process sends a $Ready$ message for $m$ (if correctly signed) to all processes in its $\tilde{\mathcal{R}}$ subscription set if it receives either:
\begin{itemize}
    \item at least $E$ $Echo$ messages for $m$ from $\mathcal{E}$; or
    \item at least $R$ $Ready$ messages for $m$ from $\mathcal{R}$.
\end{itemize}
When a correct process receives more than $D \ Ready$ messages for $m$ from its delivery sample $\mathcal{D}$ for the first time, and $m$ is correctly signed, it $prb$.Delivers $m$.
\\
\\
In the context of the Topos protocol, the probabilistic reliable broadcast alone is not enough. Indeed as mentioned above, the TCE is meant to consistently deliver messages while preserving causal order among them. In the next section we provide a definition of \textit{weak causal order} and discuss how to modify the PRB solution to extend it with a weak causal order property. Finally we show how the TCE employs Weak Causal Probabilistic Reliable Broadcast to broadcast certificates across the network, and how TCE nodes update their state accordingly.

\paragraph{Weak Causal Probabilistic Reliable Broadcast}\label{wcprb}

We now introduce the Weak Causal Probabilistic Reliable Broadcast, which extends the previous broadcast primitive with an additional weak property of causal order among the delivered messages. Intuitively, the weak causal order property imposes that if a correct process delivers a message $m$ then $m$ is weakly causally ordered with respect to the previously delivered messages.\\

For completeness, we formally recall the definition of \textit{causal precedence} \cite{modular94}:\\

\textit{Definition (Causal Precedence).} Let a step be the broadcast or the delivery of a message. A given set of steps induces a partial order as follows. Step $a$ casually precedes step $b$, denoted $a \rightarrow b$, if and only if:
\begin{itemize}
    \item the same process executes both $a$ and $b$, in that order; or
    \item $a$ is the broadcast of some message $m$ and $b$ is the delivery of $m$; or
    \item there is a step $c$, such that $a \rightarrow c$ and $c \rightarrow b$.
\end{itemize}

We now detail which kind of execution order we need to ensure for the TCE communication. In particular, we show why causal precedence is too strong for our purpose. 
Intuitively, we are interested in keeping the execution order among the messages $wcprb$.Broadcast by the same process on behalf of a subnet. We also impose the same execution order for any pair of $wcprb$.Delivered and $wcprb$.Broadcast messages by the same process. However, it is not necessary for the system to keep the execution order among messages $wcprb$.Delivered between two $wcprb$.Broadcasts from the same subnet, i.e., certificates in $deps$. For instance, as shown in Figure \ref{fig:partial_ordering}, correct processes $p_1$ and $p_2$ can $wcprb$.Deliver $Cert_{k,1}$ and $Cert_{i,2}$ in different orders.\\

To formalize the above rules, we introduce the \textit{weak causal precedence} relation, defined as follows:\\

\textit{Definition (Weak Causal Precedence).} Let a step be the broadcast or the delivery of a message. A given set of steps induces a partial order as follows. Step $a$ weakly casually precedes step $b$, denoted $a \rightarrow_{w} b$, if and only if:
\begin{itemize}
    \item the same process executes both $a$ and $b$, in that order, such that $a$ and $b$ are not a delivery step; or
    \item $a$ is the broadcast of some message $m$ and $b$ is the delivery of $m$; or
    \item there is a step $c$, such that $a \rightarrow_{w} c$ and $c \rightarrow_{w} b$.
\end{itemize}

An instance of Weak Causal Probabilistic Reliable Broadcast (WCPRB) exposes two events:
\begin{itemize}
    \item $wcprb.\text{Broadcast}(m)$: Used by a process inside the system to broadcast a message $m$;
    \item $wcprb.\text{Deliver}(m)$: Used by a process inside the system to deliver a message $m$.
\end{itemize}

For any $\epsilon \in [0, 1]$ we say that the protocol implementing the WCPRB is $\epsilon$-secure if the following properties hold: No duplication, Integrity, $\epsilon$-Validity, $\epsilon$-Consistency, $\epsilon$-Totality as defined for the probabilistic reliable broadcast, and additionally:
\begin{itemize}
    \item Weak causal order: If the $wcprb$.Broadcast of a message $m$ weakly causally precedes the $wcprb$.Broadcast of a message $m'$, no correct process $wcprb$.Delivers $m'$ unless it has previously delivered $m$.
\end{itemize}

The pseudo-code in Algorithm \ref{alg:vrb} specifies a solution for WCPRB.
This solution employs the probabilistic reliable broadcast primitive as described in the previous section and a $\texttt{Valid}$ predicate specific to the Topos system.

Before introducing the algorithm, we define the $\texttt{Valid}$ predicate. 
In the PRB solution, a message must be properly signed to be delivered. Notice that, in the Topos system, TCE processes broadcast certificates on behalf of subnets. Certificate production and submission is left to the discretion of the sending subnet. For that reason, we are not interested in identifying the process but the subnet that originated this certificate; this information is provided inside the certificate itself by means of the $\verb|subnet_id|$ field. Moreover, each certificate is $wcprb$.Delivered if the casually dependent certificates have also been $wcprb$.Delivered.\\

The message to be delivered must satisfy validity conditions. Specifically, the $\texttt{Valid}$ predicate (see Algorithm \ref{alg:valid}) is the conjunction of the two deterministic predicates $\texttt{Valid\_cert}$ and $\texttt{Valid\_deps}$:
\begin{enumerate}
    \item The certificate validation predicate $\texttt{Valid\_cert}$ (see Algorithm \ref{alg:valid_cert}) must return $\textsf{true}$;
    \item Any preceding certificate that a subnet $S_j$ issued must have been validated (implied by the linkage of certificates and encompassed in $\texttt{Valid\_cert}$);
    \item The reported dependencies of the certificate must have been validated and must exist in the histories of all subnets of the dependencies, i.e., the $\texttt{Valid\_deps}$ predicate must output $\textsf{true}$.
\end{enumerate}


The \texttt{Valid\_deps} predicate, as defined in Algorithm \ref{alg:valid_deps}, returns {\sf true} if for each certificate submitted by a subnet $S_k$ in $deps$, the certificate is in $history_{p}(S_k)$. Note that the predicate is monotonic because, for a certificate message $m$, if $\texttt{Valid\_deps}(m.deps)=\textsf{true}$ at time $t$ then $\forall t' \geq t, \texttt{Valid\_deps}(m.deps)=\textsf{true}$. The combination of the $\epsilon$-Consistency, the $\epsilon$-Totality—both providing the agreement property—and weak causal order (enforced by the $\texttt{Valid\_deps}$ predicate) properties of the WCPRB defines \textit{extrinsic validity}: All correct TCE nodes are guaranteed to deliver the same weakly causally ordered sets of certificates, i.e., no subnet can successfully submit conflicting certificates in order to double-spend.\\

\begin{algorithm}
\caption{Certificate dependencies validation predicate at process $p$}\label{alg:valid_deps}
\begin{algorithmic}

\Function{Valid\_deps}{$deps$}
    \State \Return  $\bigwedge\limits_{Cert_i \in deps} (Cert_i \in history_{p}(S_k = Cert_i.subnet\_id))$
\EndFunction

\end{algorithmic}
\end{algorithm}

When participating in the WCPRB, each TCE process $p$ locally holds the following variables:
\begin{itemize}
    \item $history_{p}(S_j)$: The local set of accepted incoming and outgoing certificates involving subnet $S_j$, for all $S_j$ (initialized and modified by Algorithm \ref{alg:certsubmission});
    \item $deps_p$: The local set of incoming certificates involving a subnet since its last outgoing certificate (initialized and modified by Algorithm \ref{alg:certsubmission});
    \item $pending_{p}$: The local set of certificates pending for validation.
\end{itemize}

Note that $deps_p$ and $history_{p}(S_j)$ are shared between Algorithm \ref{alg:valid} and Algorithm \ref{alg:certsubmission}, where only the latter modifies both variables.

\begin{algorithm}
\caption{Certificate message validation at process $p$}\label{alg:valid}
\begin{algorithmic}

\Function{Valid}{$m$}
    \State $(Cert, deps) := m$
    \State \Return $\texttt{Valid\_cert}(Cert) \wedge \texttt{Valid\_deps}(deps)$
\EndFunction

\end{algorithmic}
\end{algorithm}


The WCPRB protocol works as follows. When a TCE process wants to $wcprb$.Broadcast a message $m$ it verifies that $\texttt{Valid}(m)$ holds before calling $prb$.Broadcast($m$). Upon the $prb$.Delivery($m$) event, a TCE process $p$ does not trigger the $wcprb$.Delivery of that message yet, but adds it to the $pending_p$ set. Intuitively, $\texttt{Valid}$ being a stateful predicate, it can happen that $m$ does not satisfy the predicate at the current time but it will after the $wcprb$.Delivery of other messages. Hence, new incoming messages are kept in the $pending_p$ variable. 
As soon as there exists a message $m$ in $pending_p$ such that $\texttt{Valid}(m)$ outputs \textsf{true}, $m$ is removed from $pending_p$ and delivered.  
An interested reader can find a sketch of the proof of correctness of the proposed solution in Appendix \ref{sec:proofs}.

\begin{algorithm}
\caption{The TCE's Weak Causal Probabilistic Reliable Broadcast at process $p$ }\label{alg:vrb}
\begin{algorithmic}
\State $pending_{p} := \emptyset$
\\
\Event{$wcprb.\text{Broadcast}(m)$}
\If{$\texttt{Valid}(m)$}
\State $prb.\text{Broadcast}(m)$
\EndIf
\EndEvent
\\
\Event{$\langle prb.\text{Deliver}(m) \rangle$}
\State $pending_{p} := pending_{p} \cup \{m\}$
\EndEvent
\\
\Exists{$m$}{$\langle \{m\} \in pending_p \ \wedge \ \texttt{Valid}(m)\rangle$}
\State $pending_p := pending_p \setminus \{m\}$
\State \textbf{trigger} \ $wcprb.\text{Deliver}(m)$
\EndEvent

\end{algorithmic}
\end{algorithm}

\paragraph{Certificate Submission and State Update}\label{submission}

We detail how to transmit certificates in the Topos ecosystem using WCPRB. Algorithm \ref{alg:certsubmission} describes the submission of a certificate and its application to the local state of TCE processes.

Each TCE process $p$ locally holds the following variables:
\begin{itemize}
    \item $history_{p}(S_j)$: as defined for Algorithm \ref{alg:valid};
    \item $deps_{p}$: as defined for Algorithm \ref{alg:valid};
    \item $subnet_{p}$: The subnet that $p$ belongs to. If $p$ does not belong to any subnet then $subnet_{p}= \bot$. 
\end{itemize}

 To submit a certificate $Cert$, a subnet $S_j$ $wcprb$.Broadcasts an ICE-FROST signed message $m = (Cert, deps)$. When a correct TCE process $wcprb$.Delivers a certificate, the TCE node applies the certificate to its local state. Applying a certificate means that the TCE node adds the certificate $Cert$ and its dependencies $deps$ to the history of subnet $S_j$. More 
precisely, upon the $wcprb$.Delivery event for a message $m$ from $S_j$, a correct process $p$ always updates $history_{p}(S_j)$ but updates $deps_{p}$ only if $p$ belongs to a subnet addressed by the cross-subnet message. Notice that, if $p$ does not belong to any subnet ($subnet_{p}=\bot$) then $p$ never updates $deps_{p}$.\\

The state of the TCE (see Equation \ref{eq:state_tce}) is defined as the union of all the $history$ sets of the TCE participants. In other words, and as shown in the equation below, it is the set of all certificates that have been validated, $wcprb$.Delivered, and applied. The state of each TCE node is local and converges with probability $1 - \epsilon$.

\begin{equation}\label{eq:state_tce}
    state_{TCE}= \left\{\bigcup_{p \in \Pi} history_{p} \right\}
\end{equation}


\begin{algorithm}
\caption{The TCE's certificate submission and state update at process $p$}\label{alg:certsubmission}
\begin{algorithmic}
\State $history(S_j):= \emptyset, \forall S_j \in \mathcal{S}$  \Comment{$\mathcal{S}$: Set of all subnets}
\State $deps_{p} := \emptyset$
\State $subnet_{p} \in \mathcal{S} \cup \{\bot\}$ 
\\
\Event{submit($m$)}
\State $wcprb.\text{Broadcast}(m)$
\State $deps_{p} := \emptyset$
\EndEvent
\\
\Event{$\langle wcprb.\text{Deliver}(m) \rangle$}
\State $(Cert, deps) := m$
\State $history_{p}(S_j) := history_{p}(S_j) \cup deps \cup Cert$
\If{($subnet_{p} \in \mathcal{S}_{Cert}$)}
\Comment{$\mathcal{S}_{Cert}$: Subnets receiving messages for them in $Cert$} 
\State $deps_{p} := deps_{p} \cup Cert$
\EndIf
\EndEvent

\end{algorithmic}
\end{algorithm}

Overall, the TCE adds multiple key properties to the Topos protocol:

\begin{itemize}
    \item Security: The TCE enforces a weak causal ordering of certificates under asynchrony \cite{10.1145/167088.167105} and is a more robust primitive than consensus and atomic broadcast since both of them are impossible to solve in the asynchronous model even with one crash failure \cite{10.1145/3149.214121}.
    \item Scalability and Decentralization: With a per-node communication and computation complexity logarithmic in the size of the system, the WCPRB protocol can sustain a very high number of TCE participants—which increases decentralization—while preserving high throughput.
\end{itemize}

{\bf Discussion.} For clarity of presentation, we used the PRB protocol as a black box. However, the solution defined in Section \ref{probaRB} would not prevent the $prb$.Delivery of ill-formed messages (i.e., messages such that $\texttt{Valid\_cert}(Cert)=\textsf{false}$). Notice that this is not a problem, as those messages are not delivered by the WCPRB solution defined in Algorithm \ref{alg:vrb}. Indeed, if $\texttt{Valid\_cert}(Cert)=\textsf{false}$ then $\texttt{Valid}(m)=\textsf{false}$ (see Algorithm \ref{alg:valid}).

However, to prevent TCE processes from $prb$.Delivering messages that will never be $wcprb$.Delivered, we can move the $\texttt{Valid\_cert}$ check to the reception of new messages during the PRB protocol (see Section \ref{probaRB}). That is, at the reception of a message $m$, each correct process checks that the message is correctly signed and that the certificate carried by $m$ is well-formed ($\texttt{Valid\_cert}(Cert)=\textsf{true}$) before processing it, or discards it otherwise.  

As a consequence, at the WCPRB level (see Algorithm \ref{alg:vrb}) we can apply a lighter $\texttt{Valid}^\prime$ predicate definition, as defined in Algorithm \ref{alg:valid2}.

\begin{algorithm}
\caption{Certificate validation at process $p$, when $\texttt{Valid\_cert}$ check is integrated into the PRB solution. }\label{alg:valid2}
\begin{algorithmic}

\Function{Valid$^\prime$}{$m$}
    \State \Return \texttt{Valid\_deps}($m.deps$)
\EndFunction

\end{algorithmic}
\end{algorithm}

\paragraph{Sybil Resistance}

The TCE establishes an agreement on a set of operations among processes with equal weight. Processes are equal due to the fact that they must construct samples of peers selected uniformly at random. The TCE tolerates up to a threshold $f$ of Byzantine processes in the system. The Sybil attack \cite{10.1007/3-540-45748-8_24}, the capability of an adversary to freely create identities to overcome the $f$ threshold, is the main threat that the system has to tolerate.
Sybil resistance is relatively easy to achieve in a permissioned system contrary to permissionless systems where membership is free.
The TCE is a permissionless system and as such it is crucial to enforce that the number of Byzantine processes remain below the threshold. Notable approaches are the Proof of Work (the adversary cannot have more computational power for free) and Proof of Stake (the adversary cannot hold more assets for free). The Topos approach follows the latter, which implies the management of an asset. So in order to defend itself against Sybil attacks, the TCE leverages the Topos Subnet to ensure that a majority of reliable broadcast participants follows the protocol such that it is not possible to inconsistently deliver cross-subnet messages, i.e., double-spend.

Processes wishing to participate in the TCE must submit a special transaction on the Topos Subnet which records their intention to join the TCE. This transaction includes a fixed amount, denominated in TOPOS, as well as an identifier of the participant and it locks the TOPOS amount on the Topos Subnet. Non-free registration of participants provides the basis for a Sybil resistance mechanism in the TCE: Participants communicate only with peers registered on the Topos Subnet.



\subsection{Protocol Overview}

To this day, interoperability protocols have fallen into disjointed categories depending on their design goals. Trusted interoperability protocols (e.g., \cite{Wormhole, AvalancheBridge, PolygonPoSBridge}) have relied on external verifiers—administrators of centralized protocols or incentivized relayers in decentralized protocols—to bridge the interoperated chains, allowing for cross-chain message passing at the cost of trust in verifying entities that are external to the blockchains implementing the protocol (depending on an auxiliary, often much weaker, cryptoeconomic security). On the other hand, trustless protocols have found solutions to remove the need for trust in third parties: Some protocols (e.g., \cite{LayerZero}) have centered their design on the multiplication of non-colluding verifying networks, while others (e.g., \cite{Cosmos}) have removed the need for external verification by solely depending on the chain's own actors to natively verify data cross the chains. These models, although on the right path towards trustless interoperability, still impose trust in the consensus security of the interoperated chains and as such do not permit true trustlessness.

Topos innovates by introducing the very first solution that cryptographically enforces validity of cross-subnet messages without the need for trust in any external verifiers nor consensus security. At its core, Topos allows subnets to exchange messages with each other trustlessly and safely. \textit{Uniform security} (see Section \ref{sec:security}) is a key innovation in the cross-chain interoperability landscape and will pave the way for a new era of secure communication between decoupled blockchains.\\



Generally speaking, the Topos protocol is built on three major pillars enabled by the components exposed in the previous section.

\begin{enumerate}
    \item The TCE reliable broadcast protocol allows for consistent delivery of causally ordered subnet certificates.
    \item Certificates include a zkSTARK validity proof of the committed state transition; thus every node in the TCE network can attest to the validity of cross-subnet messages without the need to trust the sending subnet.
    \item Certificates are authenticated by means of an ICE-FROST signature; receiving subnets can thereby be ensured that delivered certificates were not tampered with.
\end{enumerate}

\subsubsection{Cross-Subnet Message}

A cross-subnet message represents a request initiated by a user from a subnet to execute a transaction in a remote subnet. It takes the form of a function call of a dedicated protocol-level smart contract, namely the Topos Core contract, on the sending subnet and is to be interpreted on the receiving subnet as another function to call. The Topos Core contract function to call on the sending subnet depends on the type of message requested:

    \begin{itemize}
        \item Asset transfer: An asset is burnt/locked on the sending subnet and equivalently minted on the receiving one.
\begin{verbatim}
    transferAsset(
        subnet_id: Identifier of the receiving subnet,
        asset_id: Identifier of the transferred asset,
        recipient_addr: Recipient's address on the receiving subnet,
        amount: Amount to be transferred
    )
\end{verbatim}
        \item Arbitrary contract call: A contract on the receiving subnet is called from the sending subnet.
\begin{verbatim}
    callArbitraryContract(
        subnet_id: Identifier of the receiving subnet,
        contract_addr: Address of the smart contract,
        func_name: Name of the function to call,
        func_args: Arguments to pass to the function call
    )
\end{verbatim}
    \end{itemize}

Topos enables interoperability of subnets via the following transmission flow of cross-subnet messages (see Figure \ref{fig:comm_flow}). Once a new cross-subnet message emitted by a user is part of the canonical chain of the subnet, it becomes ready for certification as per the rules of the UCI: it is batched with an arbitrary amount of transactions to form a new state transition whose validity is to be proven in a new authenticated certificate. Once created, the message is delivered throughout the TCE network via the reliable broadcast primitive and eventually collected by the subnet it is addressed to. Thanks to the validity and authentication properties guaranteed by the UCI and the consistent delivery ensured by the TCE, the receiving subnet can trustlessly and securely interpret the cross-subnet message and execute the request transaction locally.

\begin{figure}
\centering
\includegraphics[width=1\textwidth]{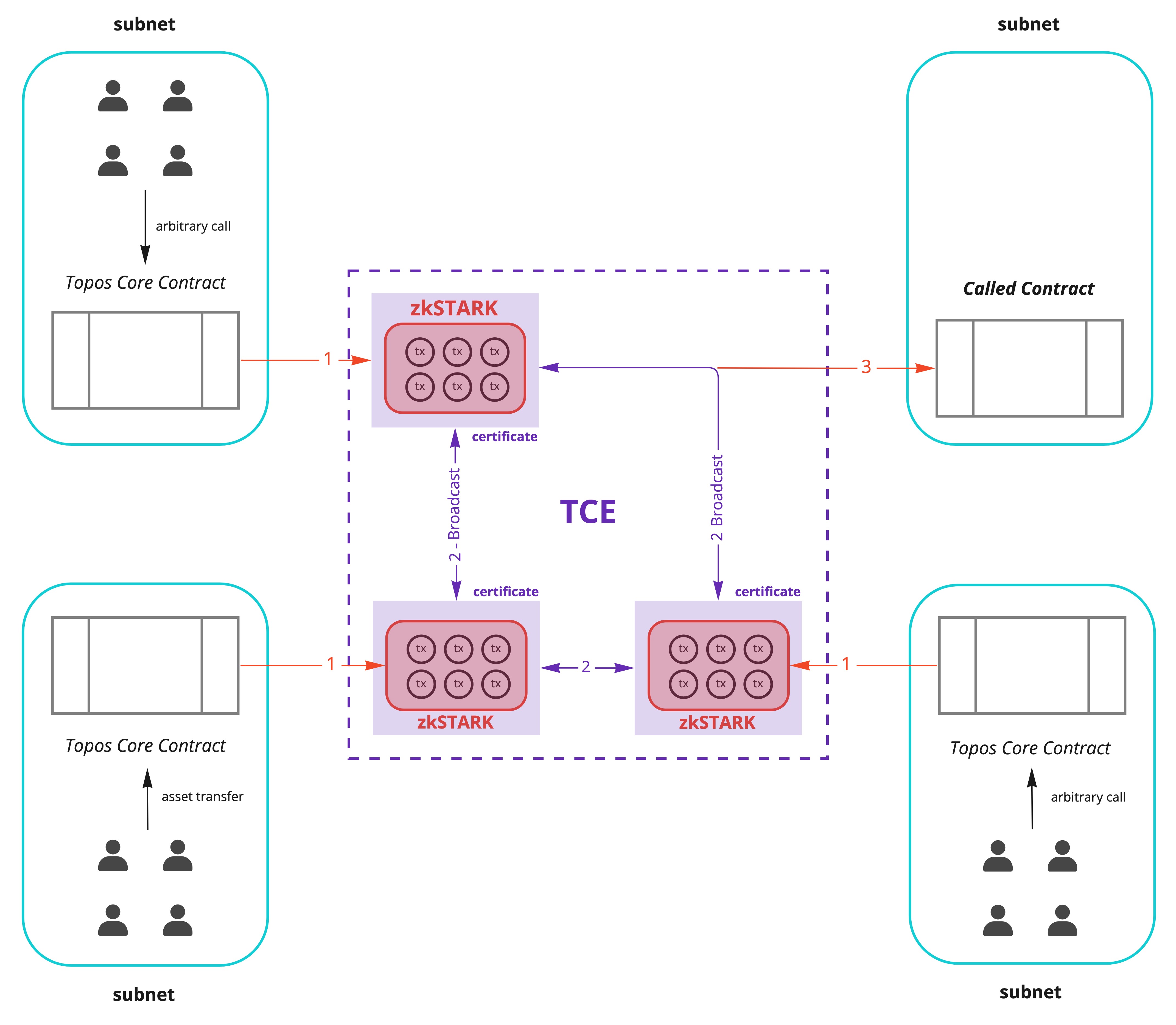}
\caption{\label{fig:comm_flow}Transmission flow of a cross-subnet message}
\end{figure}

\subsubsection{Scalability and Decentralization}

In order to verify the zkSTARK proof included in a certificate, TCE nodes do not need access to the subnet state. The zkSTARK verifier only accesses the hash of the state committed to in the previous certificate of the same subnet, and the hash of the new state committed to in the new certificate. This means that even though the size of the subnets' state transitions can be extremely large, the verification is nearly-optimal and the storage requirement for TCE nodes is kept very low. Keeping the state of the TCE small is paramount to ensure that new joining nodes can synchronize quickly and can keep the burden of storing the certificates low, even with the system processing a large amount of cross-subnet messages. While the size of the TCE state grows linearly with the number of certificates, the overhead of storing new certificates remains acceptable.

Another essential advantage of having stateless verification resides in the fact that since the cost of storing the state of the TCE is kept low, it is possible for actors with low-cost hardware to participate in the TCE, thus increasing the decentralization of the TCE.

\subsubsection{Composability}

As detailed earlier in the paper, composability is a design principle that is found when various applications can compose their value by invoking each others' functions. In the Topos ecosystem, composability is ensured in two different manners.

\paragraph{Atomic Composability}

Within a single subnet network, developers can deploy smart contracts that invoke other contracts synchronously, i.e., comprise contract-to-contract calls that are executed one after the other and only if the previous operation was successfully completed (see Figure \ref{fig:atomic_composability}). If a single operation fails, the whole transaction is reverted. In this context, composability is described as atomic for either all operations or none are executed. Commonly found in traditional databases, atomic composability allows subnets to safely transition their state and prevents them from facing corrupted state introduced by composed contract calls that fail in the middle their execution.

\begin{figure}
\centering
\includegraphics[width=0.23\textwidth]{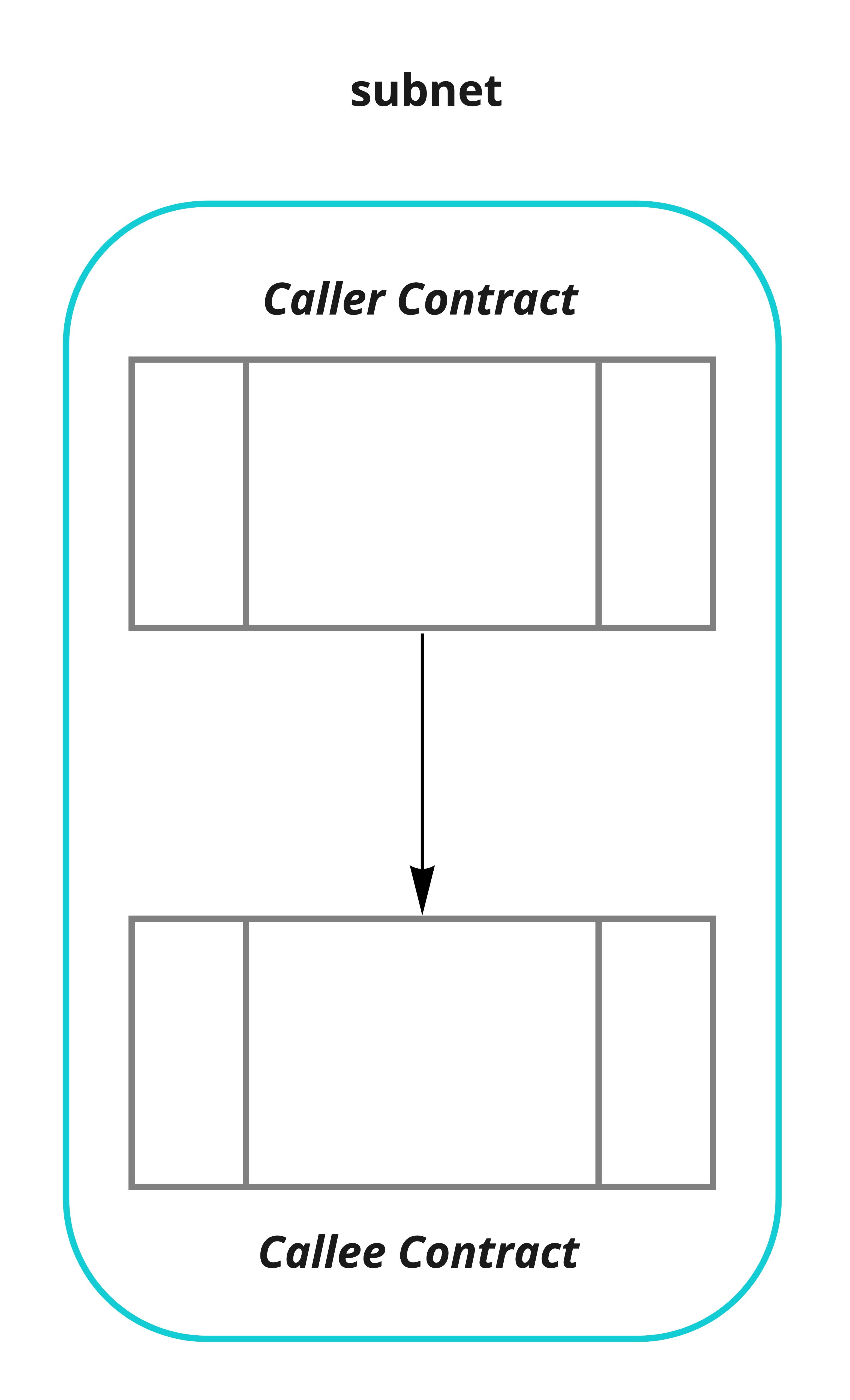}
\caption{\label{fig:atomic_composability}Atomic Composability}
\end{figure}

\paragraph{Asynchronous Composability}

In addition to atomic intra-subnet composability, the Topos protocol permits inter-subnet asynchronous composability, i.e., the capability of different applications deployed on multiple subnets to invoke each other. As we have seen previously, users of a sending subnet can emit cross-subnet asset transfers or remotely invoke arbitrary smart contracts from different subnets by calling functions of the Topos Core contract. To obtain composability across subnets, developers can atomically compose their smart contracts with the Topos Core contract, i.e., programmatically execute cross-subnet asset transfers or remote contract calls as part of their own smart contract functions. Then, subnets can include calls to these composed smart contracts in certificates for receiving subnets to learn about these new types of cross-subnet messages (see Figure \ref{fig:async_composability}).

Asynchronous composability is enabled in the Topos ecosystem by the UCI and the TCE, and is provided to any applications deployed on any subnets.

\begin{figure}
\centering
\includegraphics[width=0.8\textwidth]{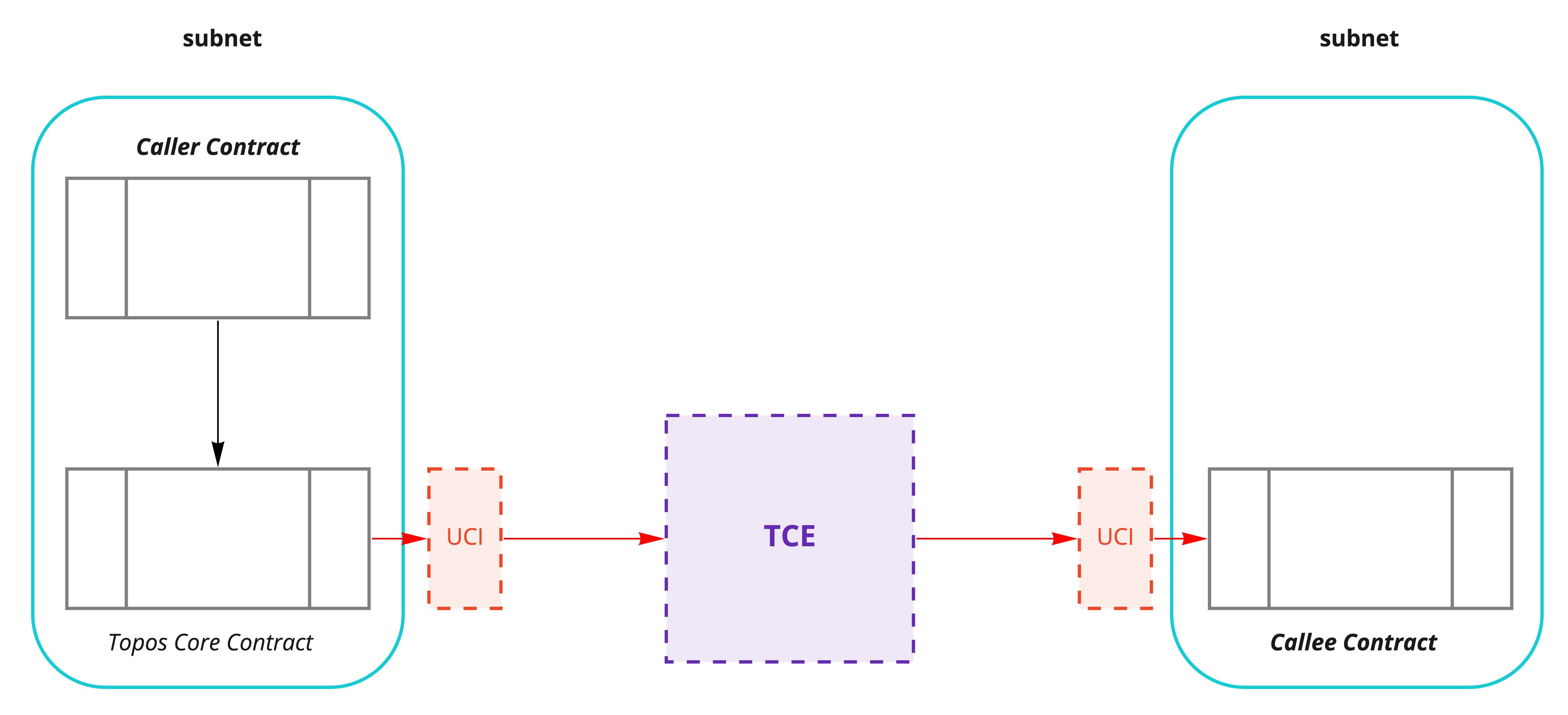}
\caption{\label{fig:async_composability}Asynchronous Composability}
\end{figure}

\subsubsection{Incentive System Design}

Topos will create economic incentives for the various actors of its ecosystem by means of its native token, TOPOS.

In blockchain systems, some actors are selfish and want to take advantage of the system. These actors deviate from the protocol, when such deviation will lead to more individual gain. As an example, they deviate to earn more rewards than when following the prescribed protocol. These actors are called \emph{rationals}.  The Topos ecosystem aims at being tolerant to Byzantine faults when the participants are rational.

In the following sections, we briefly discuss incentive system designs that will be implemented in Topos that will help align the rational actors behaviors to the prescribed behaviors. 

\paragraph{TCE Incentives}

Topos needs to prevent the verifier’s dilemma \cite{10.1145/2810103.2813659} in which, instead of maintaining the common good, correct processes choose not to verify certificates because their verification is more computationally expensive than verifying transactions, or doing anything at all.
Even if requiring the reliable broadcast participant to stake an amount of TOPOS is sufficient to provide guarantees against Sybil attacks, it is not sufficient on its own to incentivize TCE participants to follow the prescribed protocol. In fact, the system can still be subject to many predicaments, such as the verifier's dilemma. Thus, to ensure proper verification and execution of certificates, TCE participants who correctly followed the protocol should be rewarded accordingly. Such rewards come from applying a special fee to cross-subnet messages.
Fees associated to cross-subnet messages are denominated in TOPOS and are collected by TCE participants with respect to their work. Without such economic incentives, undesired situations can happen. For example, no one verifies certificates and as such invalid certificates can be spread in the system, or certificates could stay in $pending$ for a long amount of time, increasing the end-to-end cross-subnet communication latency.
Topos's incentive model will give guarantees that all certificates are processed since they potentially contain a large number of cross-subnet messages. To ensure the performance of the system, a desirable objective of the incentive model is to process the certificates as quickly as possible.
Verifying and executing certificates should be an interesting and lucrative activity incentivizing rational TCE participants to behave well, i.e., the incentive model should not reward participants who do not contribute to the system. This is ensured by the proof-of-activity mechanism introduced below.

\subparagraph{Proof-of-Activity}
Since the communication in the TCE is not synchronous, assessing whether a node does not participate is generally impossible for one cannot distinguish if a node is slow or if it is not working. The alternative approach is to prove that a node was active. Notice that this approach cannot, however, account for slow nodes. 
A proof-of-activity for a node or a set of nodes proves that their work has been seen and was considered by sufficiently many TCE participants. Since messages are signed, the proof-of-activity can consist of the set of messages delivered by a node. However, exchanging such sets of messages will induce a high communication overhead and will require too much storage on the Topos subnet; therefore, to be practical, the proof-of-activity will rely on aggregation techniques. The technical details are deferred to a subsequent paper.

\paragraph{Cross-Subnet Fee}

Cross-subnet fees must also be paid in order for the requested transaction to be executed on the receiving subnet by its validators. To estimate these fees, the transaction originator (a participant of the sending subnet) can ask a service to estimate the fee required to execute the given transaction. The service estimates and returns the result based on the fee calculation of the receiving subnet. One possibility is for this service to be exposed by a system of decentralized oracles like Chainlink \cite{Chainlink}. Internally, the estimation provider will estimate the resources consumed by processing the transaction requested in the cross-subnet message based on the current receiving subnet state.

As a general observation, if the cross-subnet fee was not greater than the one of regular transactions, the reward collected by the validator of the receiving subnet would be smaller (because shared with the TCE participants), hence it would not be in the subnet validators' best interest to process certificates in the first place.

\subsubsection{Uniform Security} \label{sec:security}

As we have seen throughout the paper, the Topos protocol allows for trustless interoperability. A sending subnet is responsible for proving that transactions are valid executions of the Topos zkVM, signing certificates that include cross-subnet messages, and sending the certificates to the TCE for broadcast in the ecosystem, while a receiving subnet is responsible for correctly applying the cross-subnet messages, i.e., submitting the requested transactions in their network.\\

Topos's \textit{uniform security} is realized by the combined properties of the UCI and the TCE:
\begin{itemize}
    \item \textit{Intrinsic validity} of certificates, ensured by the UCI's \texttt{Valid\_cert} predicate;
    \item \textit{Extrinsic validity} of certificate messages, ensured by the WCPRB's agreement property and \texttt{Valid\_deps} predicate implemented by the TCE.
\end{itemize}

\subparagraph{Intrinsic Validity.} The computational integrity of subnets' state transitions is fully decoupled from the consensus security of the related subnets and is entirely ensured by zkSTARK proofs. It is computationally infeasible for an adversary to forge a proof of validity, i.e., convince a verifier that an invalid transaction was correctly executed by the Topos zkVM. In this setup, receiving subnets benefit from an unparalleled level of security for they are assured that certified state transitions are valid. Put another way, if the UCI certificate validation predicate \texttt{Valid\_cert} outputs \textsf{true}, cross-subnet messages are guaranteed to be intrinsically valid.

\subparagraph{Extrinsic Validity.} The certificate messages that are delivered by TCE nodes are guaranteed not to be conflicting with each other while they form weakly causally ordered sets which capture the weak causal precedence between messages. This is achieved if the TCE's \texttt{Valid\_deps} predicate outputs \textsf{true}, which triggers the $wcprb$.Delivery when intrinsic validity is also verified. It follows that it is infeasible for malicious subnets to successfully double-spend and deceive honest receiving subnets into executing conflicting cross-subnet messages.\\


As soon as a certificate is delivered, \textit{intrinsic validity} and \textit{extrinsic validity} are enforced by the Topos protocol. This provides the ecosystem with \textit{uniform security}. The safety of executing cross-subnet messages internally on receiving subnets is independent of the security of sending subnets. The trust in cryptographic primitives in lieu of bridge verifiers and/or blockchain consensus as to proving the validity of the message is a fundamental innovation in the field of blockchain interoperability. It is infeasible to create a certificate containing an invalid state transition and to create conflicting certificates in order to double-spend across the ecosystem, even in the presence of a malicious subnet, e.g., more than 2/3 of its validators are controlled by an adversary. Note that transaction censorship remains an issue which needs to be addressed by subnets directly.

\subparagraph{Secure Cross-Subnet Asset Transfer.} In the practical case of a cross-subnet asset transfer, a sending subnet submits a certificate containing a proof of the validity of the asset transfer transaction. This transaction, after checking that the balance of the sender was sufficient to allow the transfer, proceeded with locking/burning the assets to be transferred. Once delivered and verified, the certificate gives total insurance to the receiving subnet that the balance check and the lock/burn operations were conducted with success on the sending subnet. This demonstrates that no malicious entity can simulate that some assets were locked/burnt in the context of a cross-subnet asset transfer, preventing them from arbitrarily minting tokens on the receiving subnet.

An additional innovation brought by the Topos protocol is in that having uniform security allows for the first trustless burn-mint asset transfer model: Receiving subnets do not need to trust sending subnets to have correctly burnt assets before minting. More secure than the lock-mint model, burning assets on sending subnets ensures that no authority can steal user assets sent across multiple chains. By enabling a trustless burn-mint asset transfer model, Topos paves the way for a new kind of asset bridge paradigm where tokens can be natively and frictionlessly deployed and managed on any blockchains.

\subsubsection{Finality and Reorganization}

With certificates being stored on the TCE, it is guaranteed that the submitted commitments to the state of subnets are immutable and hence that subnets' states cannot be reverted to states prior to the ones committed to in their latest delivered certificates without creating conflicts. In the event that a subnet reorgs to such state, the TCE would prevent the delivery of new certificates committing to a new state as these certificates would conflict with the latest delivered certificate.
This is guaranteed by the monotonicity of the TCE message predicate: If for any message $m$, $\texttt{Valid}(m)$ is $\textsf{true}$ at time $t$ then it remains true at any time $t' \ge t$. Once a certificate message is delivered via the WCPRB primitive, the certificate is considered to be final.

\section{Use Cases}\label{sec:use-cases}

The Topos protocol can suit various different use cases. Below are outlined three use cases which hint at the vast capabilities that Topos is capable of.

\subsection{Subnets as Layer-2 to Interoperate Layer-1}

As we have seen previously, in the Topos ecosystem subnets are general-purpose blockchains: They can host any type of application and prove their state transitions using the Topos zkVM, a virtual machine that allows for the execution of arbitrary provable computation. One typical use case that we envision is that of layer-2 (L2) protocols that scale layer-1 (L1) blockchain networks by delegating the execution of transactions to an offchain network and by relying on validity proofs to update the L2 state view on the L1, i.e., zk-rollups \cite{Rollup}. L2 protocols separate the \textit{execution} and \textit{settlement} layers. Traditionally, the two layers are combined: Transactions (execution layer) are executed locally by a participating node when importing a block that has been validated by the network consensus rules (settlement layer). The Topos ecosystem enables the concept of \textit{layered scalability} (see Figure \ref{fig:l1_l2}) where participating subnets are execution layers which delegate state settlement to external L1 networks. One may notice that L1 networks can greatly benefit from the instantiation of new zk-rollup subnets for they can delegate part of their execution layer and hence better scale by settling many more transactions per second. Subnets are L2 zk-rollups scaling existing secure L1s (e.g., Ethereum \cite{Ethereum}, Avalanche \cite{Avalanche}). In this configuration, settlement happens on the L1 networks where subnets publish their proofs of valid state transition.

A key consequence is that Topos indirectly enables interoperability between decoupled L1 blockchains: Any L1 chain that hosts bridge smart contracts bridging assets with a Topos subnet is de facto compatible with all other subnets in the Topos ecosystem, some of which are to be zk-rollups of other L1 networks and as such are capable to route cross-subnet messages back to their assigned L1 chain.

By design, L2 networks offer much cheaper transaction fees than L1 blockchains for they are parallel networks that do not face similar levels of congestion (if they were, the protocol operators could simply spawn a new instance of the protocol and delegate part of their execution workload to it). For that reason, dApps have been making the move to L2 and so have been their L1 users. By moving to Topos subnets, L1 users might enjoy lower fees and additionally gain for free interoperability with all other subnets in the ecosystem and in fine with all the other bridged L1 networks. The Topos protocol offers interoperability with scalability for free for existing L1s.\\

Practically, the data flow is the following (see Figure \ref{fig:comm_flow_l1}):

\begin{enumerate}
  \item To send funds to an L2 subnet, users deposit assets into a bridge contract on the underlying L1 blockchain. An equivalent number of wrapped assets is then minted on the L2 subnet for the user once the lock transaction is processed (classic \textit{lock-mint} model found in most bridge protocols).
  \item Users can initiate (wrapped) asset transfers to any other subnet in the ecosystem. They can also use their wrapped assets to pay fees when submitting requests for general-purpose computation execution, either locally on their subnet, or remotely on any other subnet (remote arbitrary transaction).
  \item A large number of subnet transactions are batched and proven by a zkSTARK proof which is transmitted to the L1 blockchain for L2 state settlement and to the rest of the Topos ecosystem via Topos certificates.
  \item The TCE reliably broadcasts certificates throughout the ecosystem. Once \textit{certificates} are delivered to all TCE nodes, the validators of receiving subnets can process the included requests for remote transaction execution and submit requested transactions locally in their subnet network.
\end{enumerate}

\begin{figure}
\centering
\includegraphics[width=0.8\textwidth]{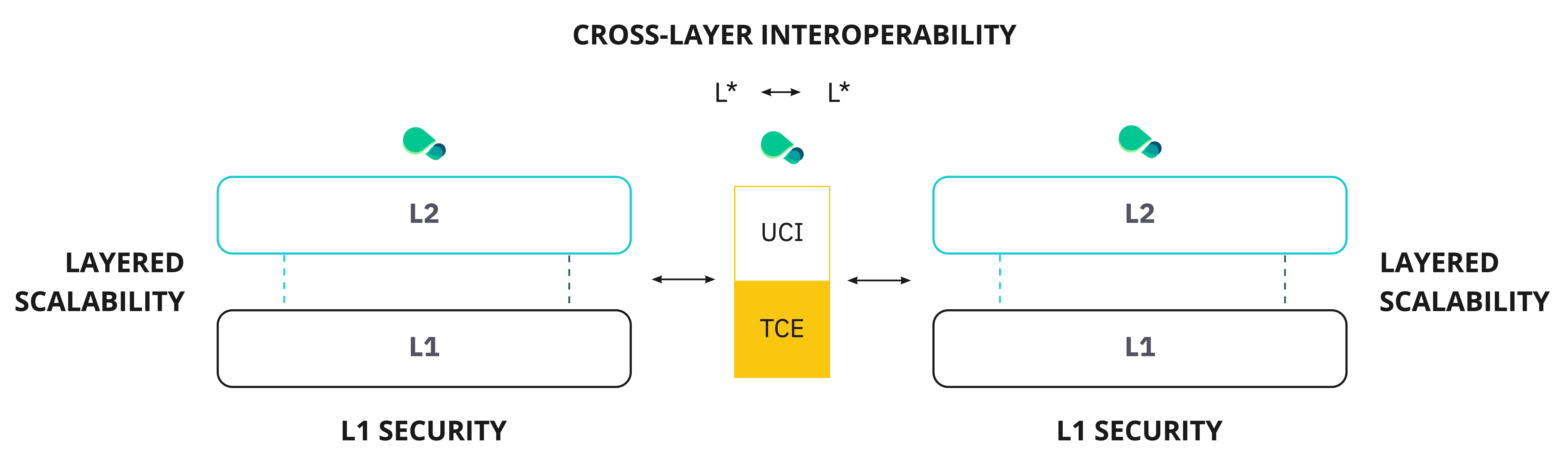}
\caption{\label{fig:l1_l2}Interaction between L1s and L2s}
\end{figure}

\begin{figure}
\centering
\includegraphics[width=1\textwidth]{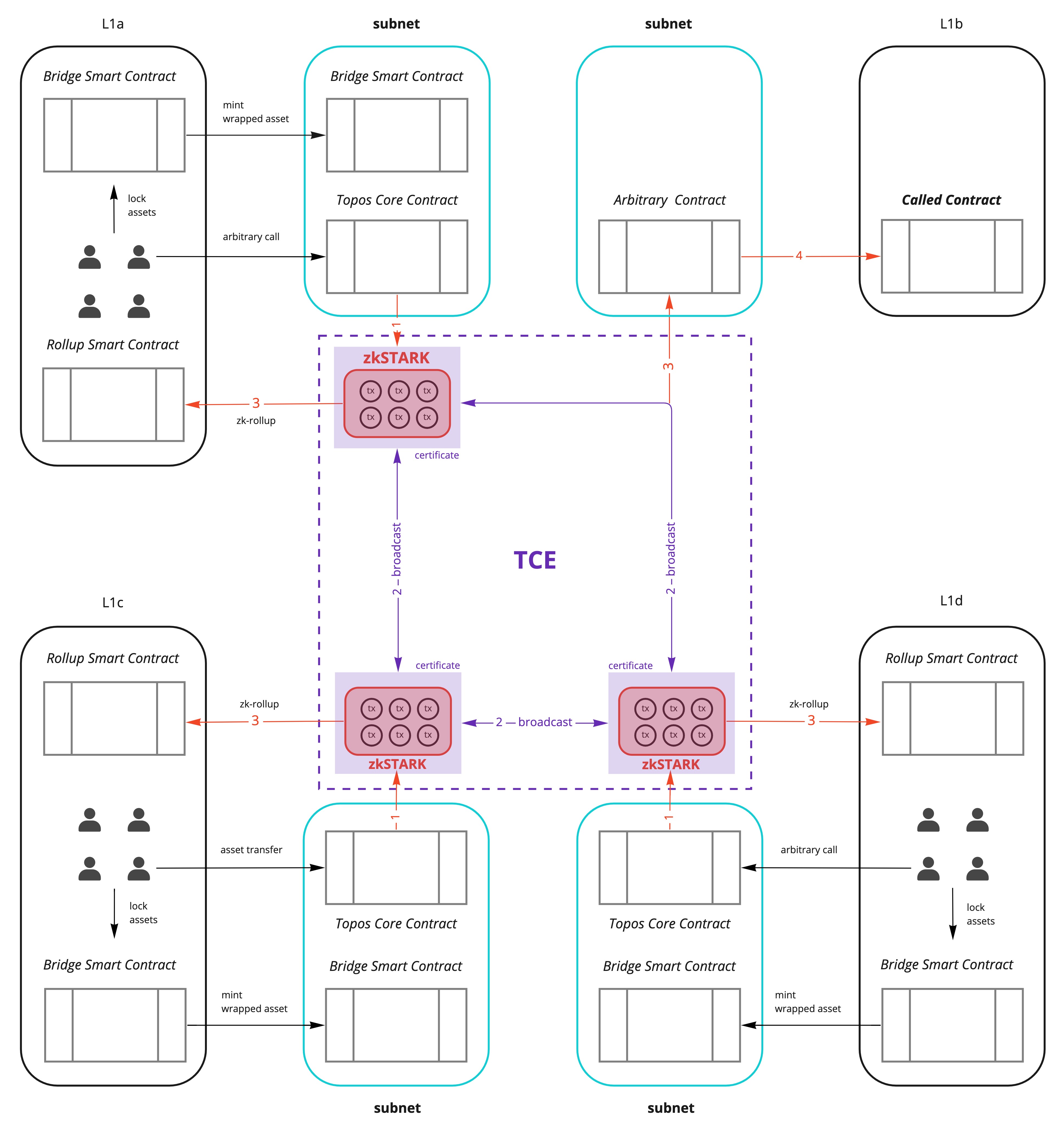}
\caption{\label{fig:comm_flow_l1}Transmission flow of a cross-subnet transaction in a zk-rollup configuration}
\end{figure}

\subsection{Decentralized Finance (DeFi)}

In the past few years, DeFi, one of the predominant use cases of web3 technologies today, has attracted a previously unseen level of user activity and commitment, with global TVL (Total Value Locked) reaching amounts in the hundreds of billions of USD. The problem is that the locked value remains isolated in silos, inaccessible for user-friendly and secure exchange options are still lacking. In most cases, liquidity, staking, loans, and other forms of DeFi value remain locked off in individual protocols, i.e., smart contracts, on individual blockchains. Due to this condition, crypto assets often lay idle in users’ wallets as there is limited scope for their use within individual blockchains. And in turn, DeFi is unable to realize its full potential.

To unlock the true potential of decentralized finance and capture massive value it is essential to allow for seamless composability of platforms and applications across different blockchain ecosystems, so that assets can move frictionlessly across different blockchain networks and be used in any DeFi protocols on any blockchains. Using trustless bridges is the best fit for DeFi because they are non-custodial, and from a security viewpoint they only leverage the security of the interoperated blockchains. Because of Topos' innovations, which give superior security guarantees than other so-called trustless interoperability protocols, value can move freely and securely from one DeFi protocol to another. The properties provided by Topos greatly enhance DeFi's potential to establish itself as a ubiquitous substitution for traditional financial systems.

\subsection{Enterprise Adoption}


The two major problems preventing enterprise adoption of blockchain technologies are the lack of interoperability and the lack of privacy \cite{ey-public:2019}.

\subparagraph{Lack of interoperability.}

Companies and organizations build value by assembling proprietary data and by controlling which part of it they expose in their products and services. Enterprises often prosper by composing value with other companies and to that end have traditionally, in the Web1-2 era, used ubiquitous infrastructure and technologies such as authenticated RESTful APIs to interface their products and services with their partners' without friction nor extra cost. Unfortunately, Web3 technologies have yet to create such standards. For enterprises to use Web3 and blockchain technologies, there is a need for a frictionless interoperability model in which blockchains running proprietary business logic and storing proprietary data can exchange and capture value openly and without the requirement to expose their private and hidden information.

\subparagraph{Lack of privacy.}

Privacy is another challenge of blockchain technologies. One of the greatest strengths of the technology is the transparency that comes from having a distributed record of transaction history that is public and easy to verify. Yet, it poses a threat to the privacy of organisations and users. Enterprises which want to protect their trade secrets and other sensitive information are therefore reluctant to embrace even the most prominent permissionless blockchain protocols and hence have favored the private/permissioned blockchain architecture.\\

Topos allows an unbounded heterogeneous set of public and private blockchains to interoperate with each other while preserving the privacy of their internal state. As such, Topos solves for both blockchain interoperability and privacy, and hence is the springboard for adoption of the technology by enterprises.

\section{Conclusion}\label{sec:conclusion}

In this paper, we introduced the design of Topos, the first interoperability protocol that is truly trustless and decentralized while preserving the privacy of the participating networks. Topos allows cross-subnet messages to be safely exchanged without the need to rely on third parties to guarantee the validity of the cross-subnet messages. Instead, we showed that this trust can be replaced by proofs of computational integrity, providing cryptographic guarantees. Another benefit of using such proofs is that verification is extremely fast and barely grows with the size of the computation which grants great throughput increase capabilities. Furthermore, a novel threshold signature scheme is introduced to facilitate authentication of messages across the Topos ecosystem. Such primitive makes it practical to manage public keys as, unlike with other threshold signature schemes, the public key remains static for the whole lifespan of subnets. In addition, we introduced a probabilistic reliable broadcast primitive to ensure consistent delivery of cross-subnet messages even in the presence of Byzantine actors. The primitive replaces the classical quorums in favor of stochastic samples while keeping the per-node communication and computation overhead very small even when the network size increases. This leads to a massively scalable and high throughput protocol.

Topos is a secure, trustless, and decentralized interoperability protocol with the aim to realize the vision of an \textit{``Internet of Blockchains"}.

\section{Discussion and Future Work}\label{sec:discussion}

\subsection{Subnet Recovery}

Subnets are sovereign networks which can be exposed to network instability, or worse, to attacks which can threaten their integrity. An example of attack is one leading a supermajority of consensus participants to be controlled by an adversary. In this situation, an arbitrary state can be finalized as part of the canonical chain. Once recovered from the attack, the subnet can start back from a valid pre-attack state only if it has access to this state. Note that on the Topos ecosystem level, the subnet can continue to participate only if there was no certificate emitted during the attack, otherwise even in the case of a recovery, the subnet would never be able to emit a certificate which is valid with respect to the subnet's history stored on the TCE.

One valuable extension of the Topos protocol can be found in leveraging data availability on the TCE to offer a recovery feature for any subnet in the ecosystem. Subnets that were compromised can query the TCE for a proof of availability, retrieve their latest committed state, and reboot from this state. 

TCE nodes store certificates which from subnets' transactional data expose only cross-subnet messages. For a data availability layer on the TCE to allow for subnet recovery, TCE nodes need to store subnets' states in clear (preventing subnets to use Topos's privacy feature). This leads to having additional data be stored, hence impacts the decentralization of the TCE network. A naive solution would be to leverage a distributed storage network, e.g., IPFS \cite{Benet2014IPFSC}, but it is not sufficient to simply store the state on it: There is need for a proof guaranteeing high availability of the storage layer and that the state is actually available. However, the successive states do not need to be available at all time for it is unnecessary to keep past data stored on that layer. Only the commitment to the state needs to be stored permanently on the data availability layer for new participants of the TCE to be able to verify previous zkSTARKs. Thereby, past states can be completely discarded, only the current state needs to be stored on the data availability layer.

Eventually, subnets can choose to include their states in certificates—giving up privacy—or not \cite{Validium} and TCE nodes only store the latest state and remove previous ones as new ones are submitted, keeping storage per subnet constant.

An alternative would be to run a data availability service on a given subnet, e.g., the Topos Subnet, with an economic incentive to store data, such that the service is well decentralized, hence minimizing trust in the data availability layer, in the sense that if the Topos Subnet is available so is the state of other subnets. A drawback of this approach is that past states cannot be discarded since they are stored on-chain.

\subsection{Confidentiality of Cross-Subnet Transactions}

Since subnets in the Topos ecosystem are sovereign blockchain networks, they can choose their level of confidentiality for internal transactions. However, for interoperability transactions between subnets must follow specific rules. In Topos's current design, the protocol specifies proofs of validity of these transactions and authenticity of the certificates they are included in but cross-subnet messages are visible in clear in the certificates. A path for future improvement is to provide confidentiality for cross-subnet messages without losing existing properties. There are a few approaches that require substantial research to achieve this goal. One such approach is to use zkSTARKs for hiding message data.

\subsection{Recursive STARKs}

As detailed in this paper, the Topos protocol introduces the UCI for subnets to speak a common language in order to, in fine, exchange data (via certificates) that is compliant with the protocol. In the current form of the protocol, certificates are created when subnets fill their batch of transactions and create a zkSTARK proving the validity of all transactions of the batch. Consequently, a subnet in which participants exchange a lot of internal transactions but submit very few cross-subnet messages will lead to the creation and propagation in the TCE of certificates that have little value in terms of interoperability in the ecosystem.

In order to relieve subnets and the TCE from this unnecessary workload, a future improvement of the Topos protocol is to leverage recursive zkSTARKs to compose proofs that are pending for insertion in certificates and ultimately create certificates only when subnets decide they have enough value in certifying some state transition (see Figure \ref{fig:stark_recursion}). Subnets will batch and prove sets of transactions as currently designed, but will accumulate the proofs and recursively compose them in a single proof for an upcoming certificate. Eventually, this will introduce flexibility and allow subnets to devise their own lineup of certificates.

\begin{figure}
\centering
\includegraphics[width=0.8\textwidth]{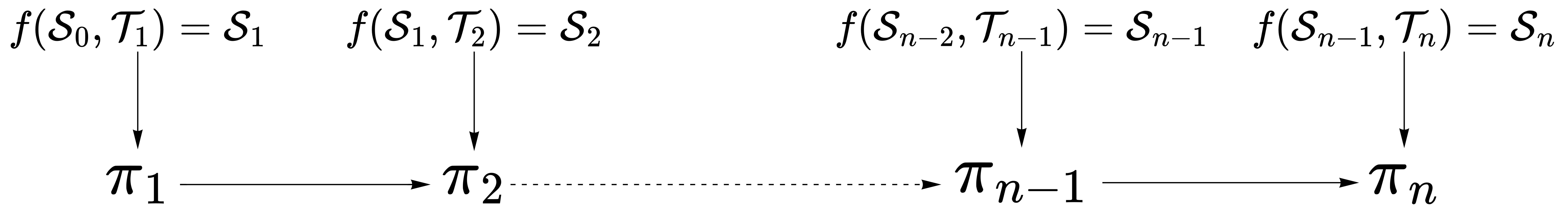}
\caption{\label{fig:stark_recursion}Recursion of STARK proofs}
\end{figure}

\subsection{Alternative to Topos Subnet}
The current Topos solution relies on the Topos Subnet for the following key operations: (i) subnets registration to the system; (ii) management of the participation in the TCE (for Sybil Attack mitigation); (iii) TOPOS asset management. Even if the Topos system security does not fully lean on the Topos Subnet, its failure would incur potential vulnerabilities in the ecosystem. For that reason we aim in the future at implementing a solution that does not depend on any dedicated subnet. This means that the three key operations described above would have to be implemented at the TCE level. Intuitively, all the information totally ordered in the Topos Subnet would have to be replicated at the TCE level, and in absence of consensus this information would appear in different order to different TCE nodes.

However, preliminary research shows that partial ordering of sets of messages is sufficient to implement all three operations. Indeed, all these operations are dependent on a token and we know that causal reliable broadcast (that the TCE implements) is sufficient to implement a cryptocurrency \cite{AT2}, even in the case where the reliable broadcast protocol implements an inflation mechanism \cite{https://doi.org/10.48550/arxiv.2105.04966}.

\section{Acknowledgements}

We would like to thank Sara Tucci-Piergiovanni and Thibault Rieutord at the CEA List and Robin Salen, Travis Baumbaugh, Alonso Gonzalez, and Hamy Ratonanina, our colleagues at Toposware for their helpful feedback and reviews which we have incorporated into the current version of the paper.

\bibliographystyle{alpha}
\bibliography{ref}

\newpage

\appendix
\section{STARK Proof System}

\subsection{Assumptions}

The only assumption required by STARK proof systems is that the hash functions to be used for commitments are collision-resistant. This allows for simpler, leaner, post-quantum and trustless proving systems than other SNARK systems.

\subsection{Prover}

The STARK proof construction can be decomposed into four stages.\\

\textbf{Algebraic Intermediate Representation (AIR).} First there is need for an algebraic representation of the problem. Consider the set $\mathcal{P}$ of the multivariate polynomials  $P_i$  in variables  $\sf{X}$  and  $\sf{Y}$,  $\mathcal{P}=\left\{P_{1}(\sf{X}, \sf{Y}), \ldots, P_{\mathrm{s}}(\sf{X}, \sf{Y})\right\}$ , where  $\sf{X}=(x_1, \cdots x_k)$ and $\sf Y=(y_1, \cdots y_k)$ represent the states of the current and next computation respectively. That is, for two correct vectors $\sf{X}, \sf{Y}$, we have that $\sf{X}, \sf{Y}$ is a correct solution for the system $\mathcal{P}$, i.e. $P_{1}(\sf{X}, \sf{Y})= \ldots = P_{\mathrm{s}}(\sf{X}, \sf{Y}) = 0$. For efficiency of both prover and verifier, we need minimal AIR, which is minimizing:
\begin{enumerate}
    \item $\deg(\mathcal{P})$=  $\max \sf \left(\deg(P_1), \cdots, \deg(P_s) \right)$;
    \item $k$, the state-width;
    \item $|\mathcal{P}|=s$;
    \item $c$, the machine cycle count (that may depend in general on arbitrary input, but is here linear in the size of the input set).
\end{enumerate}

To be able to extend and commit to the trace efficiently, the number of steps is increased so that the number of rows reaches the next power of 2, $T$, for efficient FFTs.

We often talk of execution trace for a program, which can be seen as a $T \times k$ matrix in which:
\begin{itemize}
    \item each row is describing the state of the computation at a given step;
    \item each column tracks a the content of a register over time.
\end{itemize}

\textbf{Extension of the trace and commitment.} We can view any column of the execution trace as a polynomial over a certain domain (generated by $g\in\mathbb{F}$). We can then consider the same polynomial over the domain $S=\langle\omega\rangle$, where $\omega$ is a root of $g$. This is referred to as the \textit{Low Degree Extension} \cite{10.1145/103418.103428}. The evaluation of a column polynomial $f$ on $S$ makes a code word of a Reed-Solomon code \cite{doi:10.1137/0108018} of some rate $\rho$ ($\rho = 1/\beta$ with $\beta$ the ratio between the original trace domain and the augmented LDE domain). That is, $f \in RS[\mathbb{F}, S, \rho]$. To prevent forging of proofs, it is important that the prover cannot change these values later on. Rather than sending all of these points (since we want succinctness and zero knowledge), the prover creates a \textit{commitment} to the values with a Merkle tree structure (the leaves being a grouping of evaluations of all the polynomials at a given LDE point). The commitment, along with the public inputs as part of the AIR program, is used to seed a public coin to allow drawing of random values to make the protocol non-interactive with the Fiat-Shamir heuristic.
    \begin{itemize}[leftmargin=1.5cm]
        \item If RAPs (see below) are being used, auxiliary trace segments can be computed after the previous trace portion has been committed to.
    \end{itemize}
    
\textbf{Constraint composition polynomial and consistency proof.} Similarly, the constraint polynomials $P_i$ may be composed with the column polynomials ($f_i$ for $i=1,\dots,n$) and evaluated at the points of the LDE. However, instead of creating separate evaluations for each constraint, random coins are drawn and used to create a random linear combination of the constraint polynomials. In this combination, the degrees of all constraint polynomials are augmented to all be $D$ (the next power of $2$ following the maximum degree). Due to its higher degree, more coefficients are needed to specify it. These are arranged in several columns ($H_1,\dots,H_D$), which are committed to analogously to the trace polynomials. This commitment is once again used to seed the randomness, and from this a random out-of-domain point $z$ is sampled. The prover then provides the values necessary to evaluate the constraint composition polynomial in two ways: directly through the column values $H_i$ committed to, and indirectly by evaluating the constraint polynomials at the corresponding points of the trace polynomials and performing the same linear combination as before. These values are added to the proof and the randomization of the public coin.\\

\textbf{Fast RS Interactive Oracle Proof of Proximity (FRI).}
\label{stark:fri}
The underlying idea of FRI \cite{BenSasson2017FastRI}, is to apply a similar degree-reduction to what's happening during the Inverse Fast Fourier Transform (splitting a polynomial in two instances over even and odd powers of a variable), and bind prover's responses in these reductions by evaluations of the function $f$ over points of the subset $S$.
More formally, if we can represent $S$ as $\langle\omega\rangle$ (s.t. $\omega$ generates a multiplicative group of order $2^k$) and the function $f^{(0)}: S \rightarrow  \mathbb{F}$ to be the function known by the prover of degree $d \leq \rho|S|$. The verifier will sample a random $x^{(0)}$, and ask the prover to compute $f^{(1)} = f_0 + x^{(0)} f_1$, where $f^{(1)}$ will have degree $\leq \dfrac{\rho \vert S \vert}{2} = \dfrac{\rho \vert \langle \omega^2 \rangle \vert}{2}$ for any $x$ chosen by the verifier. Here, $f_0, f_1$ are two polynomial functions such that $f^{(0)}(x)=f_0(x^2)+xf_1(x^2)$. That is, they are functions $\langle \omega^2 \rangle \rightarrow  \mathbb{F}$ with interpolants $P_{f_0}, P_{f_1}$ which are used to compute the original interpolant of $f^{(0)}$.
If $f^{(0)}$ is $\delta$-far from $\operatorname{RS}[\mathbb{F}, S, \rho]$, then the resulting $f^{(1)}$ will be $\delta'$-far for some $\delta' \leq \delta$. This process is repeated for a number of layers until the polynomial $f^{(\ell)}$ is reached which should be constant (or of low enough degree that it can be checked directly in constant time). 

If the original $f^{(0)}$ was far from any polynomial in $\operatorname{RS}[\mathbb{F}, S, \rho]$, then (with high probability) $f^{(\ell)}$ is not constant. This property of FRI is used to succinctly prove that a certain polynomial is of low degree. That polynomial is known as the \textit{DEEP composition polynomial} (from \textit{Domain Extension for Eliminating Pretenders}). It is a polynomial constructed to be of low degree only if the values previously supplied by the prover (for evaluation of the constraint composition polynomial) are consistent with the polynomials previously committed to. Due to its structure, low-degreeness of the DEEP composition polynomial also implies that the trace polynomials and column polynomials for the constraint composition are of suitably low degree.

\subsection{Verifier}

To verify a proof given by a prover, the verifier must perform the following steps.
\begin{enumerate}
    \item Read the commitment to the execution trace over the LDE domain, updating the public coin and drawing from it random coefficients used by the prover to compute the composition polynomial.
    \begin{itemize}
        \item If RAPs (see below) are being used, intermediate random coins (for use in permutation arguments) are drawn after the previous columns have been committed to.
    \end{itemize}
    
    \item Read the commitment to the constraint composition polynomial evaluations (over the LDE domain), use that to update the public coin, and sample the out-of-domain point $z$.
    
    \item Evaluate the constraints at the provided out-of-domain point $z$ based on prover-supplied trace values. Compute the evaluations of the constraint composition polynomial at the same point from the column values. Check for consistency between values. Reseeding is done after each read.
    
    \item Perform the FRI protocol: Draw coefficients for computing the DEEP composition polynomial and instantiate a FRI verifier for the layer commitments provided in the channel. Draw query positions for the LDE domain, read the evaluations of the trace and constraint polynomials at those positions. Use those to compute evaluations of the DEEP composition polynomial and verify that these are from a low-degree polynomial.
\end{enumerate}

\subsection{Randomized Air with Preprocessing (RAPs)} An additional feature desired for efficiency in STARKS is known as Randomized Air with Preprocessing (RAP). With RAPs, additional columns of the trace are committed to with access to random coins based on the original columns. This allows use of the Schwarz-Zippel lemma to show that, for instance, two columns are permutations of each other. This can be done by checking
\[\prod_{i=1}^n(a_i+\gamma)=\prod_{i=1}^n(b_i+\gamma),\]
where $a_i$ and $b_i$ are the $i$-th entry in each column, and $\gamma$ is randomly chosen after those values have been comitted to. With high probability, this only holds if the two sets $\{a_i\}$ and $\{b_i\}$ are permutations of each other \cite{Plonk}.
By supplying a known permutation $\sigma$ to the verifier, they can run a check that
\[\prod_{i=1}^n(a_i+\alpha\omega^{i}+\gamma)=\prod_{i=1}^n(a_i+\alpha\omega^{\sigma(i)}+\gamma)\]
(where $\alpha$ is randomly chosen along with $\gamma$), which indicates that $a_i=a_{\sigma(i)}$ for all $i=1,\dots,n$. This is useful for enforcing equality over great distances in the execution trace, and is referred to as \textit{copy constraints} with RAPs.

\section{ICE-FROST Signature}

The ICE-FROST protocol \cite{cryptoeprint:2021:1658} is our own adaptation of the FROST protocol \cite{frost}. The goal is to allow a subnet to generate signatures with a t-out-of-n threshold in a decentralized environment without any single trusted or semi-trusted party, and in the potential presence of malicious actors. Compared to the original FROST, our construction makes the key generation robust: enough honest actors can agree on the group's public key even in presence of malicious parties and without any rerun. In addition, honest actors can reliably identify misbehaving participants and exclude them from the scheme. At any point of time, honest actors maintain the same list of honest participants and can ignore any message from other parties.

For completeness, below is the detailed summary of the protocol.

\subsection{Assumptions}
\begin{itemize}
    \item $\mathbb{G}$ is a group of prime order $q$ in which the DDH problem is hard. $g$ is a generator of that group.
    \item The threshold $t$ and the $n$ participants are chosen by the subnet.
    \item Each participant has or receives a \textit{unique id}.
    \item Each participant has access to a $broadcast$ function. Each message published using $broadcast$ is automatically signed and available to everyone.
    \item Each participant $P_i$ is given an index $i$ between 1 and $q-1$. For simplicity we assume that the $n$ participants receive indices $1$ to $n$ but we only need them to be unique and non-zero.
    \item Honest participants want to sign a message $m$ agreed upon externally to the scheme.
\end{itemize}

\subsection{Key Generation phase}

Let $H$ be a hash function whose output is in $\mathbb{Z}_q^*$.
Let $Enc$ and $Dec$ be symmetric encryption and decryption functions.
Let $K$ be a key derivation function compatible with $Enc$ and $Dec$.\\

\textbf{Round 1}
\begin{enumerate}
    \item Every participant $P_i$ samples $t$ random values $(a_{i0},...,a_{i(t-1)}) \overset{\$}\leftarrow \mathbb{Z}_q$, and uses these values as coefficients to define a degree $t-1$ polynomial $f_i(x) = \sum_{j=0}^{t-1}a_{ij}x^j$.
    \item Every $P_i$ computes a proof of knowledge to the corresponding secret $a_{i0}$ by calculating $\sigma_i = (R_i, \mu_i)$, such that $k \overset{\$}\leftarrow \mathbb{Z}_q$, $R_i = g^k$, $c_i = H(i, \Phi, g^{a_{i0}}, R_i)$, $\mu_i = k + a_{i0}.c_i$, with $\Phi$ being a context string to prevent replay attacks.
    \item Every $P_i$ samples $sk_i$ randomly and computes $pk_i = g^{sk_i}$.
    \item Every $P_i$ computes a proof of knowledge to the secret key $sk_i$ by calculating $\tau_i = (S_i, \nu_i)$, such that $k \overset{\$}\leftarrow \mathbb{Z}_q$, $S_i = g^k$, $d_i = H(i, \Phi, pk_i, S_i)$, $\nu_i = k + sk_i.d_i$, with $\Phi$ being a context string to prevent replay attacks.
    \item Every participant $P_i$ computes a public commitment $\overrightarrow{C_i} = \langle \phi_{i0}, ..., \phi_{i(t-1)}\rangle$, where $\phi_{ij} = g^{a_{ij}}$, $0 \leq j \leq t-1$.
    \item Every $P_i$ broadcasts $\overrightarrow{C_i}$, $\sigma_i$, $pk_i$, $\tau_i$.
    \item Upon receiving $\overrightarrow{C_l}$, $\sigma_l$, $pk_l$, $\tau_l$ from participant $1 \leq l \leq n$, $l \neq i$, participant $P_i$ verifies  $\sigma_l = (R_l, \mu_l)$ by checking $R_l \overset{?}= g^{\mu_l}.\phi_{l0}^{-c_l}$, where $c_l = H(l, \Phi, \phi_{l0}, R_l)$, and $\tau_l = (S_l, \nu_l)$ by checking $S_l \overset{?}= g^{\nu_l}.pk_l^{-d_l}$, where $d_l = H(l, \Phi, pk_l, S_l)$. On failure, $P_i$ broadcasts $(malicious, P_l)$ and excludes $P_l$ from its list of participants.
    \item If the number of remaining participants for $P_i$ is below a certain value decided by the subnet, the key generation is aborted. If not, the remaining participants advance to Round 2. For simplicity, we will still refer to the remaining participants as $P_i$ ($i \in \llbracket 1,n \rrbracket$), even though some may have been eliminated at the last step of Round 1.
\end{enumerate}

\textbf{Round 2}
\begin{enumerate}
    \item Each $P_i$ does the following. For each $P_l$, $l \neq i$:
    \begin{itemize}
        \item Compute a Diffie-Hellman key $dhk_{il} = pk_l^{sk_i}$ and a symmetric key $k_{il} = K(dhk_{il})$.
        \item Encrypt $e_{il} = Enc_{k_{il}}(f_i(l))$.
        \item Broadcast $((i,l), e_{il}).$
    \end{itemize}
    \item Upon receiving $((l,i), e_{li})$ from participant $1 \leq l \leq n$, $l \neq i$, participant $P_i$ does the following:
    \begin{itemize}
        \item Compute $dhk_{li} = pk_l^{sk_i}$ and $k_{li} = K(dhk_{li})$.
        \item Decrypt $\delta = Dec_{k_{li}}(e_{li})$.
        \item Verify the share by checking $g^{\delta} \overset{?}= \prod_{k=0}^{t-1} \phi_{lk}^{i^k \text{ mod } q}$. If the share is incorrect, initiate the procedure $complain$.
    \end{itemize}
    \item Participants resolve all complaints with the procedure $exclude$. If the number of remaining participants is below a certain value decided by the subnet, the key generation is aborted. For simplicity, we will still refer to the remaining participant as $P_i$ ($i$ in $\llbracket1,n\rrbracket$), even though some may have been eliminated at the previous step.
    \item Each $P_i$ calculates their long-lived private signing share by computing $s_i = \sum_{l=1}^{n} f_l(i)$, stores $s_i$ securely, and deletes each $f_l(i)$.
    \item Each $P_i$ calculates their public verification share $Y_i = g^{s_i}$, and the group's public key $Y = \prod_{j=1}^{n} \phi_{j0}$. Any participant can compute the verification share of any other participant by calculating $Y_i = \displaystyle{\prod_{j=1}^{n} \prod_{k=0}^{t-1} \phi_{jk}^{i^k \text{ mod } q}}$. Each $P_i$ then broadcast $Y$.
\end{enumerate}

\subparagraph{complain($i,l$)}
\begin{enumerate}
    \item $P_i$ computes a proof that $dhk_{il}$ is well-formed, which is a proof of knowledge of $sk_i$ such that $(pk_i,pk_l,k_{il})$ is of the form $pk_i = g^{sk_i}$, $k_{il} = pk_l^{sk_i}$. To do so it proceeds as follow:
    \begin{itemize}
        \item It computes $A_1 = g^r$, $A_2 = pk_l^r$, where $r \overset{\$}\leftarrow \mathbb{Z}_q$ and $h = H(pk_i, pk_l, k_{il}, A_1, A_2)$.
        \item It computes $z = r + h.sk_i$.
        \item The proof is $\pi = (A_1, A_2, z)$.
    \end{itemize}
    \item $P_i$ broadcasts the message $(complaint, P_i, P_l, k_{il}, \pi)$
\end{enumerate}

\subparagraph{exclude($complaint, P_i, P_l, k_{il}, \pi$)}
\begin{enumerate}
    \item Verify the proof by checking $A_1 . pk_i^h \overset{?}= g^z$, $A_2 . k_{il}^h \overset{?}= pk_l^z$ where $\pi = (A_1, A_2, z)$ and $h = H(pk_i, pk_l, k_{il}, A_1, A_2)$. If the proof is valid, go to step 2. Else, broadcast $(malicious, P_i)$, exclude $P_i$ from the list of participants and terminate the procedure.
    \item If there is an entry $((l,i), e_{li})$ published by $P_j$, go to step 3. Else, broadcast $(malicious, P_i)$, exclude $P_i$ from the list of participants and terminate the procedure.
    \item Compute $\delta = Dec_{K(k_{il})}(e_{li})$. Verify the decrypted share by checking $g^{\delta} \overset{?}= \prod_{k=0}^{t-1} \phi_{lk}^{i^k \text{ mod } q}$. If the share is correct, broadcast $(malicious, P_l)$ and exclude $P_i$ from the list of participants. Else, broadcast $(malicious, P_l)$ and exclude $P_l$.
\end{enumerate}

\subsection{Signing phase}

We assume that a key generation phase has been successfully completed. The $n$ remaining participants now each hold a secret share, and the group's public key is $Y$. Let $H_1$, $H_2$ be hash functions whose outputs are in $\mathbb{Z}_q^*$.\\

\textbf{Round 1}
\begin{enumerate}
    \item The subnet selects randomly $S \subset \llbracket 1,n \rrbracket$, $|S| = s$, $s \ge t$ the index of signing participants. The signing participants are $P_i$, $i \in S$.
    \item Each $P_i$, $i \in S$, samples single-use nonces $(d_i, e_i) \overset{\$}\leftarrow \mathbb{Z}_q^* \times \mathbb{Z}_q^*$.
    \item Each $P_i$ broadcasts $(D_i, E_i)$ where $D_i = g^{d_i}$ and $E_i = g^{e_i}$.
\end{enumerate}

\textbf{Round 2}
\begin{enumerate}
    \item Each $P_i$ constructs $B = \langle (l, D_l, E_l)\rangle _{l \in S}$, computes the binding values $\rho_l = H_1(l,m,B)$, $l \in S$, then derives the group commitment $R = \prod_{l \in S}D_l.(E_l)^{\rho_l}$ and the challenge $c  = H_2(R,Y,m)$.
    \item Each $P_i$ computes their response using their long-lived secret share $s_i$ by computing  $z_i = d_i + (e_i.\rho_i) + \lambda_i.s_i.c$ using $S$ to determine the $i^{th}$ Lagrange coefficient $\lambda_i$.
    \item Each $P_i$ deletes $(d_i, D_i, e_i, E_i)$ from their local storage, and then broadcasts $z_i$.
    \item Each $P_i$ does the following:
    \begin{itemize}
        \item Upon receiving $z_l$ from participant $P_l$, $l \in S$, $l \ne i$, verify the validity of the response by checking $g^{z_l} \overset{?}= R_l.Y_l^{c.\lambda_l}$. On failure, broadcast $(malicious, P_l)$, exclude $P_l$ from the list of participants and go to step 5.
        \item If all responses are correct, compute the group's response $z = \sum z_i$.
        \item Broadcast the signature $\sigma = (R, z)$ along with $m$ and terminate the procedure.
    \end{itemize}
    \item If no signature has been generated and some participants have been excluded, go back to round 1 step 2 with the same $S$ minus the excluded participants. If the resulting set has less than $t$ members, abort the signature generation.
\end{enumerate}

\section{WCPRB Proof of Correctness}\label{sec:proofs}

\begin{theorem}
Algorithm \ref{alg:vrb} solves Weak Causal Probabilistic Reliable Broadcast.
\end{theorem}

\begin{proof}[Proof sketch]
To prove that Algorithm \ref{alg:vrb} solves Weak Causal Probabilistic Reliable Broadcast we prove that all the properties are satisfied. For readability we recall for each property its definition.
\begin{itemize}
    \item \textit{No duplication: No correct process delivers more than one message.} The proof follows from the No duplication property of the PRB and from the fact that when a message is $wcprb$.Delivered, it is removed from the $pending$ set.
    \item \textit{Integrity: If a correct process delivers a message $m$, and the sender $p$ is correct, then $m$ was previously broadcast by $p$.} The proof follows from the No integrity property of the PRB.
    In fact, in Algorithm \ref{alg:vrb}, to be  delivered, a message was necessary $prb$.Delivered by PRB.
    \item \textit{$\epsilon$-Validity: If the sender $p$ is correct, and $p$ broadcasts a message $m$, then $p$ eventually delivers $m$ with probability at least (1 - $\epsilon$).} The proof follows from the $\epsilon$-Validity property of the PRB and considering the following:
    \begin{itemize}
        \item If $p$ broadcasts $m$ then $\texttt{Valid}(m)=\textsf{true}$ at time $t$ and $p$ $prb$.Broadcast$(m)$. Since $\forall t' \geq t, \texttt{Valid}(m)=\textsf{true}$, then by the $\epsilon$-Validity of the PRB, $p$ $prb$.Delivers message $m$, therefore, $m$ is placed in $pending_p$ and being valid it will be removed from $pending_p$ and $wcprb$.Delivered by $p$.    
    \end{itemize}
    \item \textit{$\epsilon$-Consistency: Every correct process that delivers a message delivers
the same message with probability at least $(1 - \epsilon)$}. The proof follows from the $\epsilon$-Consistency property of the PRB.

    \item \textit{$\epsilon$-Totality: If a correct process delivers a message, then every correct process eventually delivers a message with probability at least $(1 - \epsilon)$.}
        \begin{itemize}
        \item If $p_i$ delivers $m$ then $p_i$ $prb$.Delivers $m$ and $\texttt{Valid}(m)=\textsf{true}$. If $p_i$ $prb$.Delivers $m$ then by the $\epsilon$-Totality of the PRB, every other correct process $p_j$ $prb$.Delivers $m$ and place it in $pending_j$ with probability $(1 - \epsilon)$. We now have to prove that the message $m$ in $pending_j$ will be eventually delivered, i.e.,  $\texttt{Valid}(m)=\textsf{true}$ at $p_j$. If $\texttt{Valid}(m)=\textsf{false}$ at $p_j$ at that time, it means that $p_j$ has not yet delivered the messages delivered by $p_i$ with which $\texttt{Valid}(m)=\textsf{true}$. Thanks to the $\epsilon$-Totality of PRB, $p_j$ will eventually $prb$.Deliver the same set of messages that $p_i$ $prb$.Delivered, at that time, $\texttt{Valid}(m)=\textsf{true}$, and then $p_j$ will $wcprb$.Deliver $m$.    
     \end{itemize}
\item \textit{Weak causal order:  If a correct process $p \ wcprb$.Delivers a message $m$ then $m$ weakly casually precedes all the previously $wcprb$.Delivered messages.}
 The proof simply follows from the definition of the $\texttt{Valid}$ predicate and the check on $\texttt{Valid}(m)$ before triggering $wcprb$.Deliver($m$).
\end{itemize}
\end{proof}

\end{document}